%% file: main.tex
\documentclass[onecolumn,a4paper]{quantumarticle}
\pdfoutput=1

\usepackage[utf8]{inputenc}
\usepackage{amsmath,amssymb,amsthm}
\usepackage{mathtools}
\usepackage{mathrsfs}
\usepackage{bm}
\newcommand{\ket}[1]{\lvert #1 \rangle}
\newcommand{\bra}[1]{\langle #1 \rvert}

\usepackage{graphicx}
\usepackage{booktabs}
\usepackage{pifont}
\usepackage{tikz}
\usetikzlibrary{arrows.meta,calc,decorations.pathmorphing,patterns}
\usepackage[hidelinks]{hyperref}
\usepackage{cleveref}
\crefname{condition}{Condition}{Conditions}
\Crefname{condition}{Condition}{Conditions}
\crefname{assumption}{Assumption}{Assumptions}
\Crefname{assumption}{Assumption}{Assumptions}
\crefname{estimate}{Numerical estimate}{Numerical estimates}
\Crefname{estimate}{Numerical estimate}{Numerical estimates}

\theoremstyle{plain}
\newtheorem{theorem}{Theorem}
\newtheorem{lemma}{Lemma}
\newtheorem{proposition}{Proposition}
\newtheorem{corollary}{Corollary}

\theoremstyle{definition}
\newtheorem{definition}{Definition}
\newtheorem{condition}{Condition}

\theoremstyle{remark}
\newtheorem{remark}{Remark}

\newcommand{\Zt}{\mathbb{Z}_2}
\newcommand{\anyon}[1]{\mathsf{#1}}
\newcommand{\ee}{\anyon{e}}
\newcommand{\mm}{\anyon{m}}
\newcommand{\eps}{\anyon{\epsilon}}
\newcommand{\Xbar}{\overline{X}}
\newcommand{\Ybar}{\overline{Y}}
\newcommand{\Zbar}{\overline{Z}}
\newcommand{\Hxy}{H_{XY}}
\newcommand{\Eacc}{E_{A}}
\newcommand{\Ehat}{\widehat{E}_{A}}
\newcommand{\pacc}{p_{\mathrm{acc}}}
\newcommand{\Dstab}{\Delta_{\mathrm{stab}}}
\newcommand{\Braid}{B}
\newcommand{\magicc}{\mathrm{magic}_c}
\newcommand{\Dins}{D_{\mathrm{ins}}}
\newcommand{\dFR}{d_{\mathrm{FR}}}
\newcommand{\FInst}{\mathsf{Free}}
\newcommand{\Gfrak}{\mathfrak{G}}
\DeclareMathOperator{\idop}{id}
\newcommand{\id}{\idop}
\DeclareMathOperator{\dist}{dist}
\DeclareMathOperator{\Tr}{Tr}

\begin{document}

\title{The resource cost of magic in a code block}

\author{Jiachen Shen}
\email{jshen28@cougarnet.uh.edu}
\affiliation{Department of Electrical and Computer Engineering, University of Houston, Houston, Texas 77204, USA}
\author{Hui Zhong}
\email{zhongh7@miamioh.edu}
\thanks{Corresponding author.}
\affiliation{Department of Computer Science and Software Engineering, Miami University, Oxford, Ohio 45056, USA}

\maketitle

\input{sections/abstract}
\input{sections/intro}
\input{sections/prelim}
\input{sections/results}
\input{sections/classification}
\input{sections/discussion}

\section*{Code availability}
\begin{sloppypar}
The scripts that check the numerical statements of this paper are archived at
\url{https://doi.org/10.5281/zenodo.22162937}, which resolves to the current release, and developed at
\url{https://github.com/Mercury0828/magic-cultivation-artifact}. The numbers quoted here were produced at
commit \texttt{7da798e}. Each script is self-contained, asserts the code parameters it
uses, fixes its random seed where it samples, and exits with a nonzero status on a mismatch, and
\texttt{run\_all.py} runs all six and reports one verdict. The scripts need only Python~3.9 or later
and NumPy. Where a check samples it says so and fixes its seed, while the
two statements that are not sampled, the code parameters and the rank fact behind the ambiguity group,
are asserted outright. They confirm the closed forms, the one-qubit
norm constants, and the two combinatorial facts about the rotated surface code that
\cref{app:standard,app:localclass} quote. No proof in this paper depends on them.
\end{sloppypar}

\section*{Author contributions}
J.S. and H.Z. jointly developed the framework, the theorem statements and proofs, and the manuscript. H.Z. is the
corresponding author. The authors declare no competing interests.

\bibliographystyle{quantum}
\bibliography{refs}

\appendix
\input{appendix/B-thm2}
\input{appendix/E-example}
\input{appendix/F-localclass}

\end{document}

%% file: sections/abstract.tex
% Abstract
\begin{abstract}
We bound the magic of a post-selected logical measurement by the resource that produced it. The setting is one code block with one logical qubit and an adaptive
protocol that measures, feeds forward and accepts. The witness reads the accepted effect against the free set of the resource
theory of magic, outcome by outcome and not on the averaged channel, since a channel can be free
while one of its outcomes measures the magic axis. Our first bound is unconditional. The accepted
magic is at most a constant times the summed distance of the cells from the free set. The second is the main result. When the resource cells sit inside a bounded-spread exact-recovery skeleton, the
recovery puts every insertion history below a threshold onto a single free branch, transcript by
transcript, so only connected clusters reaching the threshold contribute and an exact-component
expansion controls their weight. With a threshold linear in the code distance, polynomially many cells of bounded insertion degree and
per-cell dilation amplitude $O(1/d)$, the accepted magic times the acceptance probability is at most
$\exp[-\Omega(d\log d)]$. Post-selection is disposed of before accepted transcripts are summed, so a
branch of vanishing probability cannot be amplified into a magic effect. The threshold is certified from a circuit, and we run it
on one exact round of stabilizer measurement followed by a split readout, which measures logical $X$
on the accepted fibre and logical $Z$ on the rejected ones. That certifies a threshold equal to the code distance for every single-layer pattern
of weak $Z$-rotations, one per data qubit, so the hypotheses are met by a family and not one design. A member carries magic only if its support contains a logical $Z$ string. One
member attains the exponent, again at the level of the accepted effect. Suppression is set by the threshold and not by the
topology of the block.
\end{abstract}

%% file: sections/intro.tex
% Section 1 Introduction
\section{Introduction}\label{sec:intro}

Fault-tolerant quantum computers built from surface codes are efficient at Clifford operations and pay almost all of
their cost for the non-Clifford gates, which are supplied by magic states, of which $\ket{T}$ is the standard
example~\cite{BravyiKitaev2005,FowlerMariantoni2012}.
Fault tolerance sits alongside distributed architectures~\cite{Zhong2025UNIQ,Zhong2024DistributedQDP} and near-term
noise, error-correction and privacy
management~\cite{Zhong2023TuningQEC,Li2023ProjectionDP,Zhao2024QDPInsights,Ju2023HarnessingNoises} in the effort to
make quantum computing practical, and the layer this paper works in is the fault-tolerant one, where the magic-state
bottleneck lives.
A large-scale computation needs magic states in enormous numbers, and producing each clean one dominates the qubit and
time budget, so any saving in magic-state production is a saving in the whole machine. The standard way to produce
clean magic states is distillation, a long line of protocols that trade many noisy copies for fewer cleaner
ones~\cite{BravyiKitaev2005,BravyiHaah2012,MeierEastinKnill2013,Jones2013,CampbellTerhalVuillot2017}. Distillation is
expensive, and the reason is quantified by the resource theory of magic. That theory measures how far a state sits
from the free stabilizer set and shows that the distance cannot be created by free
operations~\cite{Veitch2014,HowardWallman2017,HowardCampbell2017,SeddonCampbell2019,BravyiGosset2016,HeinrichGross2019}.
The distance is a genuine resource, and a fault-tolerant architecture must pay to import it from outside the free set.
A recent and much cheaper alternative is magic state cultivation, which grows a $\ket{T}$ state in place on a small
surface-code patch and then enlarges the code distance~\cite{GidneyShuttyJones2024,Chen2026}. Cultivation depends on one delicate step, the
measurement of the magic axis $\Hxy=(\Xbar+\Ybar)/\sqrt2$, which is a non-Pauli logical observable whose eigenstates
are the $\ket{T}$-type states. Every surface-code cultivation performs this check through a special piece of structure, either a
fold-transversal Hadamard or a self-dual patch that turns the code back on itself~\cite{Claes2025,Sahay2025},
or a transversal operation between two blocks~\cite{Vaknin2025}.

That structure is usually described as necessary, and the reason given is topological. In the surface code each
logical Pauli is an anyon string, with $\Xbar$ carrying the charge $\mm$, $\Zbar$ carrying $\ee$, and $\Ybar$ carrying
$\eps=\ee\times\mm$. A sharp check of $\Hxy$ superposes $\Xbar$ and $\Ybar$, hence the charges $\mm$ and $\eps$, and
these two braid nontrivially with $\Braid(\mm,\eps)=-1$. The argument then runs that a plain patch supplies no local
coherent charge converter while a fold supplies one. We report it as the standard reading and use none of it.

This paper takes a different scale for the same question, because what a protocol has to pay is set by the resource
theory of magic, and the reading we develop measures each piece of a protocol against the free set of that theory. On
the stabilizer end the answer is already known, since stabilizer operations including post-selection are free, closure of the
free set under them is established~\cite{Veitch2014,HowardCampbell2017,BravyiGosset2016}, and
\cref{prop:exactnative} records what that gives on a patch. Every accepted effect of a stabilizer protocol lies in the
stabilizer-effect octahedron, for any decoder, any post-selection rule and any rate of local stochastic Pauli noise,
and every point of that octahedron is reached by a noiseless member of the class, so the reachable set of the class is
the octahedron exactly and the magic axis lies outside it. Distance, locality, the anyon labels and the exclusion of
folds do no work in the proof, which is why we treat it as a boundary marker and not as a result.

The question it leaves open is the one this paper answers. Call a cell \emph{syndrome-preserving} when its Kraus
operators commute with every stabilizer generator, so the cell creates no excitation. This is the operational form of
the statement that a cell moves no anyon, and it is decided by commutation on the lattice, not by a
homological label. A
syndrome-preserving cell may be a non-Clifford completely positive map, its branch operator need not collapse to a
projector, and the closure argument says nothing about it. Syndrome preservation is therefore not by itself a
freeness condition, and the way to constrain such a protocol is to measure how far each of its cells sits from the
free set and to ask what the total buys. Two
weak non-projective filters along different logical axes already show that the gap is real, since their composition
has accepted Bloch vector of $\ell_1$ norm above one while each cell is close to the identity.

\subsection*{Results}

We prove two bounds and exhibit a protocol that saturates the second. \Cref{fig:protocol} draws one
protocol of the kind they are about, so that the objects named below have something concrete to sit on.

\begin{figure}[t]
\centering
\resizebox{\textwidth}{!}{%
\begin{tikzpicture}[x=1cm,y=1cm,>=Latex,font=\footnotesize]

% ---- data wires ----
\foreach \y in {0,-0.5,-1.0,-1.5}{ \draw[thick] (0,\y) -- (8.5,\y); }
\draw[thick] (8.5,0) -- (8.5,-1.5);
\draw[thick] (8.5,0) -- (8.75,0.05);
\draw[thick] (8.5,-1.5) -- (8.75,-1.75);
\node[anchor=east,align=right] at (-0.15,-0.75) {one code block,\\ $d^2$ data qubits};

% ---- marked cells ----
\foreach \y in {0,-1.0}{
  \draw[thick,fill=gray!12] (1.5,\y-0.19) rectangle (2.45,\y+0.19);
  \node at (1.975,\y) {$e^{i\theta_xZ}$};
}
\draw[<-,gray!60!black] (1.975,-1.21) -- (1.975,-2.05);
\node[anchor=north,align=center] at (1.975,-2.05)
  {marked cells, at most one per data qubit,\\
   each a distance $\eta_x=2\sin(\theta_x/2)$ from the identity};

% ---- syndrome round ----
\draw[thick,fill=gray!5] (3.3,-1.8) rectangle (5.4,0.3);
\node[align=center] at (4.35,-0.75) {one exact\\ syndrome\\ round};
\draw[double,thick] (5.4,-0.75) -- (7.05,-0.75);
\node[anchor=center,fill=white,inner sep=1.5pt] at (6.0,-0.75) {record $s$};

% ---- raw-record cut ----
\draw[dashed,thick,gray!60!black] (6.55,-1.95) -- (6.55,0.75);
\node[gray!60!black,anchor=south,align=center] at (6.55,0.78)
  {raw-record cut\\ {\scriptsize classical from here on}};

% ---- branch ----
\node[draw,thick,rounded corners,align=center] (test) at (7.7,-0.75) {$s=s_{\mathrm{in}}$?};
\node[draw,thick,fill=gray!12,align=center,anchor=west] (acc) at (8.75,0.05)
  {measure $\overline X$, accept $+1$};
\node[draw,thick,align=center,anchor=west] (rej) at (8.75,-1.75)
  {apply $R(s-s_{\mathrm{in}})$, measure $\overline Z$, reject};
\draw[->,thick] (test.north) |- (acc.west) node[pos=0.62,above] {yes};
\draw[->,thick] (test.south) |- (rej.west) node[pos=0.62,below] {no};
\node[gray!60!black,anchor=south,font=\scriptsize] at (10.35,0.4) {the accepted fibre};

\end{tikzpicture}}
\caption{One member of the split-readout family of \cref{def:splitreadout}, drawn as the operations a
machine performs. Weak $Z$-rotations sit on some of the data qubits, one per qubit. A single exact round
of stabilizer measurement follows and writes its outcomes to a classical record. Everything after the
dashed cut reads that record and nothing else. If the record matches the syndrome the block came in
with, the protocol measures the logical $\overline X$ and keeps the $+1$ outcome, and that branch is the
accepted fibre whose effect this paper bounds. Otherwise it corrects, measures the logical $\overline Z$
and rejects. The rejected branch changes no accepted statistic, and \cref{prop:standard} needs it, since it is
what absorbs the leftover logical ambiguity on the other fibres.}
\label{fig:protocol}
\end{figure}
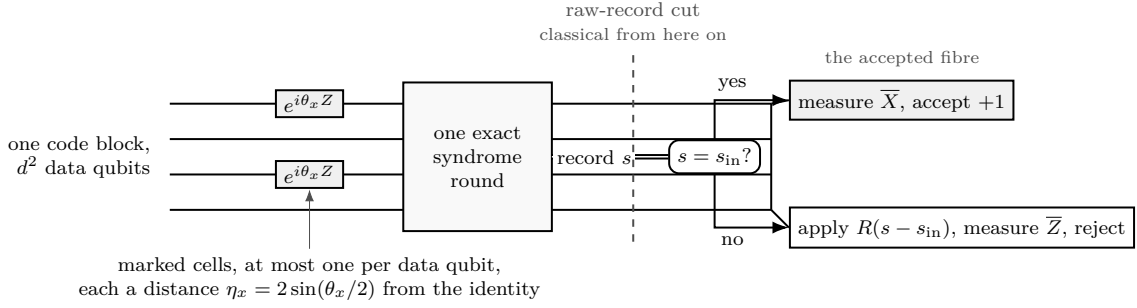

The first is unconditional, and it is written on two quantities. Writing $\varepsilon_x$ for the diamond distance of cell $x$ from the free instruments and
$\magicc(F_A)=\pacc\Dstab(\Ehat)$ for the unnormalised magic of the accepted effect, \cref{thm:resource} gives
$\magicc(F_A)\le(\sqrt3/2)\sum_x\varepsilon_x$, and \cref{cor:resource-free} replaces each $\varepsilon_x$ by the
distance of the cell to the free set itself. Read backwards, an accepted effect at constant distance from the
octahedron with non-negligible acceptance requires a constant amount of resource, distributed however the protocol
likes.

The second is the main result, and it applies when the resource cells sit inside a bounded-spread exact-recovery skeleton
on the block, and when the recovery carries every insertion history whose connected clusters stay below a threshold
$D$ onto a single free branch for each raw transcript. Such a $D$ is a \emph{certified threshold}, and
\cref{thm:suppression} bounds the accepted magic by an exact-component expansion over connected clusters of at least
$D$ cells, which \cref{cor:weak} specialises. For polynomially many cells of bounded insertion degree, a certified $D=\Omega(d)$ and per-cell
amplitude $O(1/d)$,
\[
  \pacc\,\Dstab(\Ehat)\;\le\;e^{-\Omega(d\log d)} ,
\]
so acceptance at an inverse-polynomial rate forces the magic to zero. The exponent of that bound is
$D\log(1/\eta_{\max})$, so the code distance supplies the number of cells that must cooperate and the per-cell
strength supplies the logarithm. Constant-strength cells at the same threshold give the weaker upper bound
$e^{-\Omega(d)}$ and no logarithm.

Post-selection is where the difficulty in this problem sits, and the reason it is harmless here is structural. The
branch hypothesis is imposed transcript by transcript, before accepted transcripts are summed, so an accepted branch
of vanishing probability can have its coefficient amplified without limit and it stays on its free ray while that
happens. Amplification moves a branch along the ray and does not rotate it off.

Three further pieces make the result usable, and the first of them is \cref{prop:absorption}. Given the bounded-spread skeleton and the componentwise
slabs of \cref{cond:components}, it reduces the branch hypothesis to a property of one recovery slab, that each
syndrome fibre of small insertions carries a single recovery coset up to a named logical ambiguity and that the
next free operation absorbs that ambiguity. That is what one reads off a circuit, and it replaces a check over
insertion sets.

The second is \cref{prop:standard}, which runs that criterion on the recovery round a surface-code memory already
performs. One exact round of stabilizer measurement absorbs below the $Z$-distance, with the logical ambiguity equal
to $\{I,\Zbar\}$, so a single layer of weak $Z$-rotations placed one per data qubit, on any set of qubits and at any
angles, certifies the threshold $D=d_Z$ once the split readout of \cref{def:splitreadout} follows, so
$\Dins\ge d_Z$ for all of them. It does so under named hypotheses, an exact round, an identity free skeleton, one
$Z$-diagonal cell per data qubit in a single layer, a syndrome-definite incoming branch,
terminal deviation environments and a fixed linear $Z$-recovery section,
all of which \cref{sec:results-standard} states. Coherent $Z$
over-rotation of arbitrary spatial shape is therefore inside the class, and the threshold there comes from the code
distance and is not posited. \Cref{rem:standard-tight} shows where the argument stops, at a component covering
a minimum-weight logical string.

The third is \cref{thm:construction}, which exhibits a protocol satisfying every hypothesis, one layer
of weak rotations along a bare logical string followed by syndrome extraction and a native Pauli measurement, with
$\Dins=d$, per-cell amplitude $\Theta(1/d)$, acceptance tending to $1/2$ and $-\log\magicc(F_A)=d\log d+O(d)$. The
bound is therefore attained at leading order.

\Cref{sec:results-boundary} gives two protocols outside the hypotheses that show which hypothesis does which job. A
free Clifford circuit can decode the logical qubit onto one physical wire, after which a single cell measuring $\Hxy$
produces the sharp magic projector at acceptance one half, so a small geometric footprint is not a small resource and
the threshold must be measured against the recovery skeleton. And $r$ commuting rotations at each site of a string
contribute $r^d$ insertion histories that all evaluate to the same logical operator, so their amplitudes add in phase
and many weak cells can act as one strong one when they accumulate between recovery rounds.

\subsection*{Scope}

The suppression is a code-distance statement and nothing in its proof is specific to the toric-code phase. The
distance enters through the threshold $\Dins$ and the cluster count enters through the degree bound, while the
braiding relation is never used, so \cref{thm:suppression} reads verbatim on any stabilizer code carrying one logical
qubit whose recovery satisfies \cref{cond:skeleton,cond:degree,cond:h4}. The free set there is the one-qubit
octahedron and the norm conversions are three-dimensional. Which codes
satisfy the recovery hypotheses is a question about recovery circuits, not about topology, and it is open. The anyon picture of the opening paragraphs is motivation for the problem and is not
part of any statement below.

The threshold enters as a named hypothesis, \cref{cond:dins}, and \cref{cor:weak} is a statement about protocols that
meet it. Two results supply it. \Cref{prop:absorption} gives the circuit-level criterion by which a threshold
is certified for a given recovery skeleton, and \cref{prop:standard} runs that criterion on one exact round of
stabilizer measurement and certifies $D=d_Z$ for a single layer of $Z$-diagonal single-qubit deviations placed
anywhere on the block, followed by the split readout of \cref{def:splitreadout}. The same question for cells with support on several
qubits, for cells off the $Z$ axis, and for resource that accumulates across several rounds is open. \Cref{sec:classification} places each known construction by the first hypothesis it leaves.
Every construction there that reaches a constant amount of magic leaves one, and \cref{sec:discussion} says which.

The main theorem is stated on the dilation amplitudes $\eta_x$ and not on the diamond distances $\varepsilon_x$.
Continuity of the Stinespring dilation~\cite{KretschmannSchlingemannWerner2008} relates the two for some choice of
dilations, and \cref{rem:amplitude} records the compatibility the substitution needs, which is automatic for unitary
cells and is a hypothesis otherwise.

\Cref{sec:prelim} fixes the block, the description of an adaptive protocol as an accepted effect, and the magic
witness. \Cref{sec:results} contains the results. \Cref{sec:classification} reads them against the known
constructions and \cref{sec:discussion} sorts what carries what. \Cref{app:thm2} proves the two bounds,
\cref{app:example} works the distance-3 case, and \cref{app:localclass} proves the exact reach of the stabilizer class
and the properties of the weak-string protocol.

%% file: sections/prelim.tex
% Section 2 Preliminaries
\section{Setup and definitions}\label{sec:prelim}

This section fixes the object we study and states the definitions the results rest on. We work with one fixed
piece of hardware, a single planar surface-code patch that stores one logical qubit and is never deformed. On this
patch we ask how much magic an adaptive post-selected protocol can accept, and what its non-stabilizer cells have to
pay for it. To make that question precise we need three ingredients. First, we fix the code, its logical operators
and the measurement we want to perform (\cref{sec:prelim-block}). Second, we describe an arbitrary adaptive,
post-selected protocol as a single accepted effect (\cref{sec:prelim-effects}). Third, we define a witness that tells
a genuine magic measurement apart from a classical imitation (\cref{sec:prelim-witness}). We close with the fault
model (\cref{sec:prelim-fault}), which enters one result only, \cref{prop:replacement}.

We use standard surface-code facts without reproving them, and the stabilizer formalism describes the code and its
logical operators~\cite{Gottesman1997,AaronsonGottesman2004,DehaeneDeMoor2003}, while the toric-code phase and its
gapped boundaries describe the anyons, the boundaries, and the edge condensates~\cite{Kitaev2003,KitaevKong2012,BeigiShorWhalen2011,Barkeshli2019}. Protection and the acceptance decision
rest on the topological-memory and decoding literature~\cite{Dennis2002,DuclosCianciPoulin2010,DelfosseNickerson2021}.
Lattice surgery~\cite{Horsman2012,RaussendorfHarringtonGoyal2007,LitinskiVonOppen2018,Litinski2019} and the
dynamics of Floquet and dynamical codes~\cite{HastingsHaah2021,Vuillot2021,Davydova2023} are the standard ways to move
logical information between blocks, and they are outside the fixed single-patch setting fixed here.
\Cref{tab:notation} collects the symbols used throughout.

\begin{table}[t]
\centering
\caption{Notation used throughout.}
\label{tab:notation}
\resizebox{\textwidth}{!}{%
\begin{tabular}{ll}
\toprule
Symbol & Meaning \\
\midrule
$d$ & code distance of the planar surface-code patch \\
$\{1,\ee,\mm,\eps\}$ & toric-code anyons, $\eps=\ee\times\mm$, with braiding $\Braid$ and $\Braid(\mm,\eps)=-1$ \\
$\Xbar,\Ybar,\Zbar$ & logical Paulis, Wilson strings of charge $\mm,\eps,\ee$ \\
$\Hxy=(\Xbar+\Ybar)/\sqrt2$ & target magic axis, with L\"uders projectors $(I\pm\Hxy)/2$ \\
$K_\tau,\ E_\tau=K_\tau^\dagger K_\tau$ & Kraus operator and effect of a fine transcript $\tau$ \\
$A,\ \Eacc=\sum_{\tau\in A}K_\tau^\dagger K_\tau$ & accepted family and its coarse-grained effect \\
$P=P_{\text{code}},\ F_A=P\Eacc P$ & code projector and the code-space part of the accepted effect \\
$\pacc=\tfrac12\Tr F_A,\ \Ehat=F_A/(2\pacc)$ & acceptance probability and normalised logical effect \\
$\Dstab(\Ehat)$ & magic witness: distance of $\Ehat$ to the stabilizer-effect octahedron \\
$\magicc(F_A)=\pacc\Dstab(\Ehat)$ & unnormalised magic of the accepted effect \\
$\FInst_x,\ \mathbf G$ & free flagged instruments at cell $x$, and the fixed free skeleton \\
$\varepsilon_x(\mathbf G)$ & diamond distance of cell $x$ to its skeleton comparison \\
$\eta_x$ & dilation amplitude of cell $x$, the norm of $V_x-W_x$ \\
$\mathcal R_d,\ m=|\mathcal R_d|$ & marked cells and their number \\
$\Gfrak_d,\ \zeta$ & insertion-dependence graph and its degree bound \\
$\Dins$ & insertion distance: cluster size below which recovery restores a free branch \\
$L_F(\delta),\ \dFR$ & fault distance of the accepted effect at scale $\delta$, and free-replacement distance \\
\bottomrule
\end{tabular}}
\end{table}

\subsection{The surface-code block and the target measurement}\label{sec:prelim-block}

We fix a planar rotated surface code encoding one logical qubit, with fixed rough and smooth boundaries and code
distance $d$~\cite{Kitaev2003,BravyiKitaev1998,FowlerMariantoni2012,BombinMartinDelgado2007}. Data qubits sit on the
faces of a rotated square lattice, and the stabilizer group is generated by weight-four $X$-type and $Z$-type checks in
the bulk, with weight-two checks along the boundaries. The two rough boundaries absorb $\ee$ charge and the two smooth
boundaries absorb $\mm$ charge, and this choice of condensates is what pins the logical operators to definite anyon
type. The distance $d$ is the least weight of a logical operator, and it is the length of the shortest string
connecting a pair of like boundaries.

The bulk realises the toric-code phase $D(\Zt)$, whose anyons are $\{1,\ee,\mm,\eps\}$ with $\eps=\ee\times\mm$. The
braiding form $\Braid$ is nondegenerate, and the only nontrivial values we use are
$\Braid(\ee,\mm)=\Braid(\mm,\eps)=\Braid(\ee,\eps)=-1$, while every charge braids trivially with itself and with the
vacuum.
The logical Pauli operators are Wilson lines of definite anyon type. $\Xbar$ is an $\mm$-string between the two smooth
boundaries, $\Zbar$ is an $\ee$-string between the two rough boundaries, and $\Ybar=i\Xbar\Zbar$ is an $\eps$-string
along the diagonal. The single logical anticommutation $\Xbar\Zbar=-\Zbar\Xbar$ is the lattice image of the toric-code
braiding $\Braid(\ee,\mm)=-1$. Every use of this fact below is through the anticommutation, which is a property of the
logical Pauli group, and never through the braiding form.

The logical measurement we target is the nondemolition L\"uders check of the magic axis
\begin{equation}\label{eq:hxy}
  \Hxy=\frac{\Xbar+\Ybar}{\sqrt2},\qquad
  \text{with L\"uders projectors}\quad \frac{I\pm\Hxy}{2}.
\end{equation}
This is the check that a surface-code cultivation of the $\ket{T}$ state needs~\cite{GidneyShuttyJones2024,Claes2025,Sahay2025}.
The hard part is coherence, because an $\Hxy$ eigenstate keeps a fixed phase between the $\mm$-string sector and the
$\eps$-string sector, and a classical choice between them is not enough. The physical picture usually offered for why
this should be hard is that $\Braid(\mm,\eps)=-1$, so the two sectors braid nontrivially and a coherent $\Hxy$ effect
asks the patch to carry both charges at once. We record that picture as motivation and the results below do not use
it. What they use is the logical anticommutation $\Xbar\Zbar=-\Zbar\Xbar$, the free set of the resource theory of
magic, and the code distance, and \cref{sec:discussion} returns to the difference. \Cref{app:example} works the whole
story on the smallest patch, the distance-3 code, and is a concrete instance to keep in mind while reading the
definitions.

\subsection{Adaptive protocols as accepted effects}\label{sec:prelim-effects}

A cultivation protocol measures, adapts, and then accepts or rejects. It interleaves measurements, classical
feed-forward, and a final acceptance decision, and only accepted runs produce output. We describe it at two grains, of which the finer is the
single \emph{fine transcript} $\tau$, one complete record of measurement outcomes and feed-forward choices. It acts
on the code space by a Kraus operator $K_\tau$, and its accepted effect is $E_\tau=K_\tau^{\dagger}K_\tau$, the positive
operator for which $\langle\psi|E_\tau|\psi\rangle$ is the probability that the run follows transcript $\tau$ and is
accepted, starting from $|\psi\rangle$. The \emph{raw transcript} $y$ of \cref{sec:results} is the measurement record alone, and $\tau$ is $y$
together with the feed-forward choices the record determines, so the two label the same branch.
A decoder groups fine transcripts into an \emph{accepted family} $A$, the set
of transcripts the protocol declares successful. The object the reader should keep in mind throughout is the
coarse-grained effect of that family,
\begin{equation}\label{eq:accepted-effect}
  \Eacc=\sum_{\tau\in A}K_\tau^{\dagger}K_\tau ,
\end{equation}
We abbreviate the code projector as $P=P_{\text{code}}$ and write
\begin{equation}\label{eq:FA}
  F_A\;=\;P\,\Eacc\,P
\end{equation}
for the code-space part of the accepted effect, which is the object every bound below is stated on. Its acceptance
probability on the maximally mixed logical input is $\pacc=\tfrac12\Tr F_A$, which is where the factor $\tfrac12$
comes from, and its normalised logical effect is $\Ehat=F_A/(2\pacc)$. We must work with $\Eacc$, not
any single $K_\tau$. A decoder can collapse exponentially many transcripts into one high-acceptance event, and a
statement proved transcript by transcript can fail for the family they form.

One more point fixes the level at which we work. A protocol defines a full quantum instrument, the map that sends the
input to the accepted output together with the classical acceptance flag. The instrument carries more information than
its effect $\Eacc$. It records the post-measurement state, whereas $\Eacc$ records only the accepted outcome
statistics, so the effect does not by itself determine the output state. We prove an obstruction at the level of the
effect. This is a \emph{necessary condition} for the instrument: a nondemolition $\Hxy$ L\"uders measurement has a
definite accepted effect, the projector $(I+\Hxy)/2$, so ruling out that effect rules out any instrument that would
realise the measurement. The obstruction is a statement about accepted-outcome statistics, which is the level at
which the witness below is posed. It captures exactly the part a
coherent check must get right, and the witness below acts at the same level.

\subsection{The magic witness}\label{sec:prelim-witness}

Not every effect with $\Xbar$ and $\Ybar$ support is a magic measurement, so we need a witness that is not fooled
by classical imitations. The right object is the distance of the normalised logical effect from the set of effects
that native Pauli measurements can build. For one logical qubit, write $\Ehat=\tfrac12(I+\bm v\cdot\bm\sigma)$ with
Bloch vector $\bm v=(v_x,v_y,v_z)$. The native operations are the six Pauli-Wilson-measurement projectors
$\tfrac12(I\pm\Xbar)$, $\tfrac12(I\pm\Ybar)$, $\tfrac12(I\pm\Zbar)$, whose Bloch vectors are the six unit vectors
$\pm\hat e_x,\pm\hat e_y,\pm\hat e_z$. A protocol that measures Pauli axes and mixes the outcomes classically produces
a convex combination of these, so its normalised effect has $\bm v$ in their convex hull, the stabilizer-effect
octahedron $\{|v_x|+|v_y|+|v_z|\le 1\}$. This octahedron is the one-qubit face of the resource theory of
magic~\cite{BravyiKitaev2005,Veitch2014,HowardCampbell2017}, the same polytope of free effects whose extreme points are
the stabilizer measurements. An effect outside it cannot be assembled from Pauli measurements and classical mixing, so
a positive distance to it is a certificate that the check does something no classical combination of Pauli readouts
can.

\begin{definition}[Magic witness]\label{def:dstab}
The \emph{magic} of a normalised one-qubit logical effect $\Ehat=\tfrac12(I+\bm v\cdot\bm\sigma)$ is its
$\ell_1$-distance to the stabilizer-effect octahedron,
\begin{equation}
  \Dstab(\Ehat)=\dist\!\bigl(\bm v,\ \{\,\bm w:\|\bm w\|_1\le1\,\}\bigr).
\end{equation}
\end{definition}

\begin{figure}[t]
\centering
\begin{tikzpicture}[scale=2.2,>=Latex]
  % axes
  \draw[->] (-1.35,0)--(1.45,0) node[right] {\footnotesize $v_x$};
  \draw[->] (0,-1.35)--(0,1.45) node[above] {\footnotesize $v_y$};
  % octahedron cross-section (v_z=0): diamond |v_x|+|v_y|<=1
  \draw[thick,fill=gray!12] (1,0)--(0,1)--(-1,0)--(0,-1)--cycle;
  \node[gray!60!black,align=center] at (-1.0,-0.95) {\footnotesize stabilizer\\octahedron};
  \draw[gray!60!black,thin] (-0.72,-0.78)--(-0.42,-0.42);
  % classical mixture point (1/2,1/2) on the edge
  \fill (0.5,0.5) circle (0.02);
  \node[anchor=west] at (0.53,0.46) {\footnotesize $\tfrac12P_{\Xbar,+}\!+\!\tfrac12P_{\Ybar,+}$};
  % H_XY point outside
  \fill[red] (0.707,0.707) circle (0.022);
  \node[red,anchor=south west] at (0.72,0.72) {\footnotesize $\Hxy$};
  % Delta_stab gap: double arrow between mixture point and H_XY; label sits just outside the diamond, upper-left of the arrow
  \draw[red,<->] (0.5,0.5)--(0.707,0.707);
  \node[red] at (0.40,0.92) {\footnotesize $\Dstab$};
\end{tikzpicture}
\caption{The $v_z=0$ slice of the stabilizer-effect octahedron $\{\|\bm w\|_1\le1\}$, whose points $\bm v$ are the
normalised logical effects $\tfrac12(I+\bm v\cdot\bm\sigma)$ that Pauli measurements and classical mixing can
build, with $\Dstab$ the $\ell_1$-distance to it. The projector onto the $+1$ eigenspace of the magic axis
$\Hxy=(\Xbar+\Ybar)/\sqrt2$ (red) lies outside at $\Dstab=\sqrt2-1$. The classical mixture
$\tfrac12P_{\Xbar,+}+\tfrac12P_{\Ybar,+}$ of the two logical Pauli projectors sits on the boundary edge, with
$\Dstab=0$.}
\label{fig:octahedron}
\end{figure}

The witness does exactly the job we need. The octahedron $\{\|\bm w\|_1\le1\}$ has eight facets, one for each choice of
signs $(\pm,\pm,\pm)$, and in the octant of the sharp check the nearest facet is $w_x+w_y+w_z=1$. The $\ell_1$-distance
of a point $\bm v$ outside it is the excess of the $\ell_1$ norm over one,
\begin{equation}
  \Dstab\bigl(\tfrac12(I+\bm v\cdot\bm\sigma)\bigr)=\max\bigl(0,\|\bm v\|_1-1\bigr).
\end{equation}
The sharp $\Hxy$ projector has $\bm v=(1/\sqrt2,1/\sqrt2,0)$, hence
\begin{equation}
  \|\bm v\|_1=\tfrac{1}{\sqrt2}+\tfrac{1}{\sqrt2}=\sqrt2>1,\qquad \Dstab=\sqrt2-1>0,
\end{equation}
so it lies strictly outside the octahedron (\cref{fig:octahedron}). A coin flip between an $\Xbar$ measurement and a
$\Ybar$ measurement gives $\tfrac12P_{\Xbar,+}+\tfrac12P_{\Ybar,+}$ with $\bm v=(\tfrac12,\tfrac12,0)$ and
$\|\bm v\|_1=1$, so it sits on the octahedron boundary and $\Dstab=0$. The two effects have the same Pauli support,
along $\Xbar$ and $\Ybar$, and the witness still tells them apart, because it sees the $\ell_1$ length of the Bloch
vector, which the set of axes present does not determine, and this is the property the first condition
will name.

Distances between logical operators are measured in the \emph{coefficient norm}. Writing any operator on the logical
qubit in the Pauli basis as $A=a I+\bm v\cdot\bm\sigma$ with $a\in\mathbb C$ and $\bm v\in\mathbb C^3$, we set
$\|A\|_c=|a|+\|\bm v\|_1$, which is a norm on the whole four-dimensional operator space and restricts to real
coefficients on Hermitian operators. The two scales differ by a factor of two, since $\Ehat$ carries the
prefactor $\tfrac12$ and so $\dist_c$ between normalised effects is $\tfrac12\Dstab$, while \cref{lem:norm} works with the
unnormalised effect, where the factor is absorbed into $\pacc$. At the unnormalised level the free set is the
cone
\begin{equation}\label{eq:freecone}
  \mathcal N=\{\,q(I+\bm w\cdot\bm\sigma):q\ge0,\ \|\bm w\|_1\le1\,\} ,
\end{equation}
so $F\in\mathcal N$ exactly when its normalised effect lies in the octahedron, and \cref{lem:norm} identifies
$\magicc(F_A)$ with the coefficient-norm distance $\dist_c(F_A,\mathcal N)$. We use both forms below. The witness is built from the same free set as
the resource theory of magic. The stabilizer-effect octahedron is the convex hull of the native Pauli-measurement
effects, and $\Dstab$ is the $\ell_1$ distance of the Bloch vector to it, so $\Dstab$ vanishes exactly
on that free set and is positive only for an effect no classical mixture of Pauli measurements can build. It is the
effect-side analogue of the state quantities that measure how far a state is from the stabilizer
polytope~\cite{Veitch2014,HowardCampbell2017,SeddonCampbell2019,HeinrichGross2019}. We use it as a witness and order
parameter, and the proofs draw on exactly two of its properties. It vanishes on the free set and nowhere else, and it
is a distance to that set, so it obeys the triangle inequality. Monotonicity under free operations
is a further property of the state-side quantities that the arguments here never invoke, and
\cref{sec:classification-dynamical} says how that places the witness against the channel monotones of the resource
theory. That separation is what makes it non-circular. It is defined by the native operations, not by the conclusion we want to reach, and a protocol
whose accepted effect lands outside the octahedron has demonstrably done more than measure and mix Pauli axes.

\subsection{The fault model}\label{sec:prelim-fault}

One later result, \cref{prop:replacement}, compares a protocol against faults, this subsection fixes the fault model
it uses, and nothing else in the paper depends on it.

The adversary may replace the Kraus operators of any set of spacetime cells by arbitrary ones, and may flip any
\emph{physically represented} classical bit, meaning a raw measurement record bit or a location of the physical
classical processing that acts on those bits, one cell or bit per fault. The decoder and the acceptance function are
evaluated at the boundary of the fault model, so the decoded logical outcome and the acceptance flag are values
computed from the raw record, and not single-bit fault locations of their own, which is what makes the count below a
property of the protocol. A record produced once and then copied gains nothing from the copying, since a fault
upstream of the fan-out corrupts every copy alike, so a record that survives $\ell$ faults is one the protocol
produces redundantly, by repeating the underlying logical measurement or through a fault-tolerant parity interface,
in such a way that the decoded outcome depends on $\Theta(\ell)$ independently faulted raw bits.

Write $F_A^{\mathrm f}$ for the code-space part of the accepted effect of the faulted protocol.

\begin{definition}[Fault distance of the accepted effect]\label{def:LF}
For $\delta>0$, the \emph{fault distance at scale $\delta$} is
\begin{equation}\label{eq:LF}
  L_F(\delta)\;=\;\min\bigl\{\,|\mathcal F| \;:\; \mathcal F \text{ a fault set with } \|F_A^{\mathrm f}-F_A\|_c\ge\delta\,\bigr\} .
\end{equation}
A protocol is \emph{protected to $\ell$ at scale $\delta$} when $L_F(\delta)\ge\ell$.
\end{definition}

Three points fix how \cref{eq:LF} is to be read. The comparison is on the unnormalised $F_A$, so a fault set that
suppresses the acceptance probability without changing the logical action counts as a change, and a fault set that
destroys acceptance outright counts too. The faults are counted across the whole instrument, so selectors, acceptance
logic and decoding are included alongside the code cells. And the scale $\delta$ is carried explicitly and not fixed
once, because the bound of \cref{prop:replacement} holds at every scale up to $\magicc(F_A)$ and says nothing
above it.

Protection in this sense is a property of the protocol against the adversary, and it is not a hypothesis of
\cref{thm:resource} or \cref{thm:suppression}. Those theorems constrain the accepted effect of a specified protocol
and the adversary plays no part in their proofs. \Cref{sec:results-fault} is where the two meet, and
\cref{sec:discussion} says why they should not be read as evidence for one another.

%% file: sections/results.tex
% Section 3 Main results
\section{The cost of magic on a code block}\label{sec:results}

This section fixes what a code block can accept when its non-stabilizer resource is weak and
spread out. Two results bound the accepted magic by the resource that produced it, and a third exhibits a
protocol that saturates the second at leading order.

The objects are three. The \emph{free} operations are the outcome-wise completely stabilizer-preserving
instruments of \cref{sec:results-free}, which is the free set of the resource theory of magic read at the
level of a single measurement outcome, not of an averaged channel. The \emph{marked} cells are
those that are not free, and each carries a strength. The \emph{insertion distance} of
\cref{sec:results-class} counts how many marked cells must act together before the accepted branch can
leave the free set. The theorem runs at any threshold that has been certified, \cref{cond:dins} asks for
one that grows with the distance, and \cref{prop:absorption} is how a circuit certifies one.

\Cref{tab:deps} records what each result needs, so that a reader can see at a glance which statements
are unconditional and which carry hypotheses.

\begin{table}[t]
\centering
\caption{What each result rests on, by name. Every result is stated on one encoded logical qubit, which
the free set and the norm conversions both use, so that hypothesis is not repeated. The paragraph after
the table spells out the entries that are not a numbered condition.}
\label{tab:deps}
\resizebox{\textwidth}{!}{%
\begin{tabular}{lll}
\toprule
Result & Rests on & Proved in \\
\midrule
\Cref{lem:norm} (normalisation) & nothing & \cref{app:thm2} \\
\Cref{prop:exactnative} (exact reach) & stabilizer branches & \cref{app:localclass} \\
\Cref{lem:free-native} (free branches are native) & a stabilizer-cone Choi operator & \cref{sec:results-free} \\
\Cref{thm:resource} (resource accounting) & any fixed free skeleton & \cref{app:thm2} \\
\Cref{thm:suppression} (suppression) & a compatible certificate: \cref{cond:skeleton,cond:degree,cond:h4} & \cref{app:thm2} \\
\Cref{cor:approx} (approximate factorisation) & \cref{thm:suppression} and \cref{eq:approx-h4} & \cref{app:thm2} \\
\Cref{cor:weak} (weak distributed resource) & \cref{thm:suppression} at a certified $D\ge\beta d$ & \cref{sec:results-theorems} \\
\Cref{prop:absorption} (absorption gives \cref{cond:h4}) & \cref{def:absorption,cond:skeleton,cond:components} & \cref{sec:results-criterion} \\
\Cref{prop:standard} (a standard round certifies $d_Z$) & the setting of \cref{sec:results-standard} and \cref{def:splitreadout} & \cref{app:standard} \\
\Cref{prop:family-magic} (the family carries magic) & \cref{def:splitreadout} with unitary cells & \cref{app:standard} \\
\Cref{thm:construction} (saturation) & \cref{def:weakstring} & \cref{app:localclass} \\
\Cref{prop:replacement} (replacement bounds protection) & the fault model of \cref{sec:prelim-fault} & \cref{sec:results-fault} \\
\bottomrule
\end{tabular}}
\end{table}

Three entries need unpacking. \Cref{prop:exactnative} also asks that the noise be a Pauli mixture and
that the class be closed under classical mixing, the latter only for its converse.
\Cref{lem:free-native} asks that the outcome subchannel, composed with the encoding isometry, have its
Choi operator in the stabilizer cone. And the setting of \cref{sec:results-standard} is a CSS code with
an identity free skeleton, one exact L\"uders syndrome round, $Z$-diagonal deviations placed one per
qubit in a single layer, a syndrome-definite incoming branch, terminal deviation environments and a fixed
linear $Z$-type recovery section.

\subsection{The stabilizer boundary}\label{sec:results-stab}

One case is exactly soluble and fixes the free set concretely on the patch. When every branch of the
protocol is a stabilizer operation, the accepted effect can be computed exactly.

\begin{proposition}[Exact reach of the stabilizer class]\label{prop:exactnative}
Let a protocol family on the distance-$d$ patch have cells drawn from local Clifford and Pauli
operations, native Pauli--Wilson measurements, product stabilizer ancillas and classical feed-forward,
with local stochastic Pauli noise at rate $p$, and let the class be closed under classical mixing. For
every accepted set with $\pacc>0$ and at every rate $p$,
\begin{equation}\label{eq:exactnative}
  F_A\in\mathcal N,\qquad\text{equivalently}\qquad\Dstab(\Ehat)=0 .
\end{equation}
Conversely every normalised effect with $\Dstab=0$ is the accepted effect of a noiseless member of the
class, with $\pacc\ge1/2$ uniformly in $d$. The reachable set is exactly the stabilizer-effect
octahedron.
\end{proposition}

\Cref{app:localclass} proves it, and the forward half is three steps. A stabilizer branch has
$K^\dagger K=c_0\Pi_S$ for a commuting group $S$, the projector collapse of \cref{lem:collapse} turns
$P_{\text{code}}\Pi_SP_{\text{code}}$ into a nonnegative multiple of a further stabilizer projector, and
a stabilizer projector on one logical qubit is $I$, a native projector $\tfrac12(I\pm\bar P)$, or zero,
while local stochastic Pauli noise mixes those with nonnegative weights and $\mathcal N$ is a convex cone.
The converse builds any octahedron point from a classical mixture of native measurements, and the two
halves meet on the same generators.

Read as a resource statement, \cref{prop:exactnative} is the closure of the free set of the resource
theory of magic under stabilizer operations with post-selection~\cite{Veitch2014,HowardCampbell2017,BravyiGosset2016,BravyiKitaev2005},
specialised to an encoded qubit, and the same proof applies to any $[[n,1]]$ stabilizer code. Its role here
is to fix the boundary concretely, so that the later bounds have something to measure against. Distance, locality, the anyon labels and the exclusion of folds
are unused in its proof, which is the first sign that the topological reading of this problem does less
work than it appears to.

What \cref{prop:exactnative} leaves open is the interesting part. A cell may be a non-Clifford completely
positive map and still be a plausible ingredient of a static-patch protocol, and for such a cell
$K^\dagger K$ need not be a projector, so the collapse argument says nothing. Two weak non-projective
filters on different logical axes already illustrate the gap. Write $A=I+a\Xbar$ and $B=I+b\Zbar$, which
become valid instrument outcomes after division by $1+|a|$ and $1+|b|$. Those constants scale the accepted
effect and cancel from its normalised Bloch vector, so we suppress them and compute with $A$ and $B$
themselves. Then
\begin{equation}\label{eq:twofilter}
  (BA)^\dagger(BA)=(1+a^2)(1+b^2)\,I+2a(1+b^2)\,\Xbar+2b(1-a^2)\,\Zbar ,
\end{equation}
whose normalised Bloch vector has $\ell_1$ norm $1.28$ at $a=b=\tfrac12$, and which approaches $\sqrt2$
as $(a,b)$ vary. At $a=b=1$ both filters become projective, the collapse applies again, and the norm
returns to exactly $1$.

Call a cell \emph{syndrome-preserving} when every Kraus operator of every one of its outcomes commutes
with every stabilizer generator of the patch, so the cell creates no excitation and leaves every syndrome
bit untouched. This is the operational content of the statement that a cell moves no anyon, and unlike
the anyon statement it is decided by commutation on the lattice. Both filters above are
syndrome-preserving, since $\Xbar$ and $\Zbar$ are logical operators and commute with the whole
stabilizer group, and their composition has accepted effect outside $\mathcal N$. Syndrome preservation
is therefore not a freeness condition, and it is not the quantity a bound can be stated on. The rest of
this section measures a cell by how far it sits from the free set instead.

\subsection{Free instruments and marked cells}\label{sec:results-free}

A protocol is a finite sequence of instruments, each applied at one spacetime location of the circuit, and
each such location is a \emph{cell}. The decomposition into cells comes with the protocol, and every
quantity below is read against it. An adaptive protocol applies instruments whose choice depends on the
classical record so far, and we record that dependence explicitly. For a cell $x$ with instrument $\{\mathcal M_{o|h}\}_o$ controlled by the
classical history $h$, the \emph{flagged} channel is
\begin{equation}\label{eq:flagged}
  \widetilde{\mathcal M}_x\bigl(\rho\otimes\ket{h}\bra{h}\bigr)
  \;=\;\sum_o \mathcal M_{o|h}(\rho)\otimes\ket{h,o}\bra{h,o} ,
\end{equation}
an ordinary channel on the system together with a classical register in a fixed basis. Flagging is what
keeps the outcome label visible, and the label is where the resource hides. A channel can be free while
one of its outcomes is not, and the example that settles this is
$\mathcal M_\pm(\rho)=\Tr(E_\pm\rho)\ket{0}\bra{0}$ with $E_\pm=(I\pm\Hxy)/2$. Its sum over outcomes is
the constant preparation $\ket{0}\bra{0}$, which is free, while the outcome statistics perform the magic-axis
measurement. Freeness must therefore be imposed outcome by outcome.

\begin{definition}[Free flagged instrument]\label{def:free}
A flagged instrument is \emph{free} when every outcome subchannel $\mathcal M_o$ is completely
stabilizer-preserving, meaning that $(\mathcal M_o\otimes\id_r)$ maps subnormalised mixtures of
stabilizer states to subnormalised mixtures of stabilizer states for every reference size $r$. Write
$\FInst_x$ for the free flagged instruments at cell $x$.
\end{definition}

Free instruments cannot produce a magic accepted effect, and the proof is one line of the Choi
correspondence.

\begin{lemma}[Free branches are native]\label{lem:free-native}
Let $\mathcal J$ be an outcome subchannel whose composition with the encoding isometry $V_L$ has
subnormalised Choi operator in the stabilizer cone. Then $V_L^\dagger\mathcal J^\dagger(I)V_L\in\mathcal N$.
\end{lemma}

\begin{proof}
Apply $\mathcal J$ to half of the encoded Bell stabilizer state $\ket{\Phi}$ and write
$\Omega=(\mathcal J\otimes\id)(\ket{\Phi}\bra{\Phi})$, which lies in the stabilizer cone by hypothesis.
Tracing out the output gives $\Tr_{\text{out}}\Omega=\tfrac12(V_L^\dagger\mathcal J^\dagger(I)V_L)^T$.
Partial trace and transpose both preserve the stabilizer cone, and on one logical qubit that cone is
$\mathcal N$.
\end{proof}

\begin{remark}[A Kraus operator is not enough]\label{rem:kraus-not-enough}
\Cref{lem:free-native} asks that the subchannel be free, not that its Kraus operator belong to some free
channel. Kraus decompositions are not outcome decompositions, and a free channel can be written with
Kraus operators that individually are not free, so the hypothesis has to be placed on the subchannel.
\end{remark}

The distance from a cell to the free set is measured against \emph{one fixed} free comparison. Fix a
\emph{free skeleton} $\mathbf G=\{\widetilde{\mathcal G}_x\}_{x}$ with $\widetilde{\mathcal
G}_x\in\FInst_x$ and set
\begin{equation}\label{eq:eps}
  \varepsilon_x(\mathbf G)\;=\;\bigl\|\widetilde{\mathcal M}_x-\widetilde{\mathcal G}_x\bigr\|_\diamond .
\end{equation}
Everything below, including the insertion expansion and the insertion distance, is read against the same
$\mathbf G$. This matters, and \cref{rem:skeleton} shows what goes wrong when it is ignored. The
\emph{marked} cells are $\mathcal R_d=\{x:\varepsilon_x>0\}$, with $m=|\mathcal R_d|$.

\subsection{The insertion expansion and the class}\label{sec:results-class}

Purify every cell and every classical record, so that the whole adaptive protocol is one isometry $V$ on
the system, the record registers and a dilation environment. Adaptive control becomes coherent control by
the record registers, and no generality is lost. Choose common dilations
\begin{equation}\label{eq:dilation}
  V_x\;=\;W_x+\Delta_x,\qquad \eta_x=\|\Delta_x\| ,
\end{equation}
with $V_x$ dilating $\widetilde{\mathcal M}_x$ and $W_x$ dilating $\widetilde{\mathcal G}_x$. Expanding
every marked slot gives the \emph{insertion expansion}
\begin{equation}\label{eq:insertion}
  V\;=\;\sum_{S\subseteq\mathcal R_d}V_S ,
\end{equation}
where $V_S$ carries $\Delta_x$ at the marked slots in $S$ and $W_x$ at the others. For a complete raw
transcript $y$ write $K_{y,S}=\bra{y}V_S$. The single fact the proof needs about \cref{eq:insertion} is
the cylinder bound of \cref{lem:cylinder}, and its content is that forced non-insertions are free.

The class is defined by five conditions. The first fixes the setting, the next two are structural, the
fourth is the working hypothesis, and \cref{sec:results-criterion} derives that fourth one from properties
of a circuit.

\begin{condition}[One protected algebra]\label{cond:one-block}
The protocol runs on one code block. No ancillary register carries a second logical algebra whose
distance grows with $d$.
\end{condition}

\Cref{cond:one-block} fixes the setting and is not a hypothesis of any bound below. The proof in
\cref{app:thm2} never invokes it, and it does not need to: a second protected algebra that is inert
changes nothing, and one that is used in a way the bound would otherwise miss has to break
\cref{cond:h4}, which is already a hypothesis. It is stated because it names what the threshold is measured against,
one block whose distance is the $d$ in $\Dins=\Omega(d)$, and because it is the first thing the
multi-block constructions of \cref{sec:classification} give up.

\begin{condition}[Bounded-spread recovery skeleton]\label{cond:skeleton}
Marked cells are inserted into a skeleton of bounded-range free operations and exact recovery on the same
block. Each recovery slab ends at a \emph{raw-record cut}, the moment at which every check ancilla has
been measured and its outcome stored in an orthogonal classical register. The free operations of the
skeleton carry each syndrome sector into a syndrome sector, so that the block enters every slab with a
definite syndrome.
\end{condition}

\begin{condition}[Bounded insertion degree]\label{cond:degree}
The \emph{insertion-dependence graph} $\Gfrak_d$ on $\mathcal R_d$ joins $x$ and $x'$ when the forward
light cones of $\Delta_x$ and $\Delta_{x'}$, each stopped at the raw-record cut of its slab, meet in a
recovery cell, or when the supports of $x$ and $x'$ lie within a fixed range $r_\ast$. Its degree is at
most a constant $\zeta$, independent of $d$.
\end{condition}

The cut in \cref{cond:degree} is the whole point. Classical decoding, the acceptance predicate and every
later outcome-dependent free operation happen after it and contribute no edges. Taking instead the
transitive causal closure through the global acceptance computation would join every marked cell to every
other through the classical aggregation circuit and give degree $\Theta(d)$, which is why the graph is
stopped where it is. Enlarging $\Gfrak_d$ by further bounded-range edges is harmless, since it only
shrinks the set of insertion sets the next condition speaks about.

The next condition is the working hypothesis, and in words it asks that below a threshold the recovery put
every insertion history of a given transcript on one free logical ray, leaving the scalar and the
environment vector free.

\begin{condition}[Branch factorisation]\label{cond:h4}
There is a threshold $\Dins(d)$ such that for every complete raw transcript $y$ and every insertion set
$S$ all of whose connected components in $\Gfrak_d$ have fewer than $\Dins(d)$ cells,
\begin{equation}\label{eq:h4}
  \bigl(\bra{y}\otimes I_E\bigr)V_S\,P\;=\;\bigl(J_yP\bigr)\otimes\ket{e_{y,S}}_E ,
\end{equation}
where $J_y$ does not depend on $S$, the environment vector $\ket{e_{y,S}}$ may, and the logical
subchannel generated by $J_y$ satisfies the hypothesis of \cref{lem:free-native}.
\end{condition}

\Cref{eq:h4} is written on the environment vector and not on a scalar, which makes it invariant under
a common change of basis inside a fixed dilation. The ray $J_y$ is the same for every insertion set below the threshold, which is what makes the sum over
those sets stay free.

A positive integer $D$ for which \cref{eq:h4} holds is a \emph{certified threshold}, and the condition is
monotone downwards, since shrinking $D$ only shrinks the set of insertion sets it speaks about. The
\emph{insertion distance} $\Dins$ is the supremum of the certified thresholds, which is $\infty$ when
every insertion set factorises, in which case \cref{eq:lambda} is an empty sum and the bound below is
$\magicc(F_A)=0$. \Cref{thm:suppression} holds at every certified threshold and is strongest at the
largest one, so a result that certifies some $D$ gives the bound at that $D$ together with
$\Dins\ge D$, and that is the form every certification below takes.

\begin{condition}[Distance-scale threshold]\label{cond:dins}
$\Dins(d)=\Omega(d)$.
\end{condition}

The quantities the bound is written on are read against a choice, and the following names that choice so
that the bound can be stated on the protocol and not on the choice.

\begin{definition}[Compatible insertion certificate]\label{def:certificate}
A \emph{compatible insertion certificate} for a protocol is a tuple
$\mathcal C=(\mathbf G,\{V_x,W_x\},\Gfrak_d,D)$ consisting of a free skeleton, a choice of common
dilations \cref{eq:dilation}, an insertion-dependence graph and a threshold, such that
\cref{cond:skeleton,cond:degree,cond:h4} hold for the protocol with those objects and with that $D$. Its \emph{weight} is
\begin{equation}\label{eq:certweight}
  \Lambda(\mathcal C)\;=\;\sum_{\substack{C\subseteq\mathcal R_d\ \mathrm{connected}\\ |C|\ge D}}\ \prod_{x\in C}\|V_x-W_x\| ,
\end{equation}
and the \emph{certified cost} of the protocol is $\Lambda^\star=\inf_{\mathcal C}\Lambda(\mathcal C)$ over
its compatible certificates, with $\Lambda^\star=\infty$ when it has none.
\end{definition}

\Cref{thm:suppression} is stated at one certificate, and it does not bound the accepted magic by
$\Lambda(\mathcal C)$ itself. Writing
\begin{equation}\label{eq:fofl}
  f(\lambda)\;=\;\sqrt3\,\min\bigl\{1,\ 2(e^{\lambda}-1)+3(e^{\lambda}-1)^2\bigr\} ,
\end{equation}
the theorem gives $\magicc(F_A)\le f(\Lambda(\mathcal C))$ for every compatible $\mathcal C$. The function
$f$ is continuous and nondecreasing, rising from $f(0)=0$ and reaching its ceiling $\sqrt3$ at
$\lambda=\log(4/3)$, where the two branches of the minimum agree, and staying there. When the protocol has
at least one compatible certificate, passing to a minimising sequence gives
\begin{equation}\label{eq:lambdastar-bound}
  \magicc(F_A)\;\le\;f(\Lambda^\star) ,
\end{equation}
and in the linear range $\Lambda^\star\le\log2$ this reads $\magicc(F_A)\le10\sqrt3\,\Lambda^\star$. When
it has none, $\Lambda^\star=\infty$ and \cref{eq:lambdastar-bound} still holds, now as the universal
bound $\magicc(F_A)\le\sqrt3=f(\infty)$, which follows from $0\in\mathcal N$ and $0\le F_A\le I$ and needs
no certificate. That
infimum is a property of the protocol alone, and it is the object the title's phrase names. The individual
$\eta_x=\|V_x-W_x\|$ are not, since they are dilation amplitudes read against a chosen skeleton in a
chosen environment gauge, and \cref{rem:amplitude,rem:skeleton} say what goes wrong when they are treated
as intrinsic.

\begin{remark}[The skeleton is one object]\label{rem:skeleton}
The quantity $\varepsilon_x$ of \cref{eq:eps} is the distance to a chosen $\widetilde{\mathcal G}_x$, and
it is not the distance to the free set $\FInst_x$. The two differ, and the whole bound must be optimised
over skeletons jointly, not by minimising the strengths and maximising the threshold separately,
because the two extrema are attained at different skeletons. The family of
\cref{sec:results-construction} at angle $\theta_d=\arctan(\tan(\pi/8)^{1/d})$ shows the gap. Against the
identity skeleton it has $\varepsilon_x=\Theta(1)$ and $\Dins=d$, while against the skeleton of logical
$S$ rotations it has $\varepsilon_x=O(1/d)$ and $\Dins=1$.
\end{remark}

\subsection{The two bounds}\label{sec:results-theorems}

The first bound needs none of \cref{cond:one-block,cond:skeleton,cond:degree,cond:h4,cond:dins}. It says
that magic is paid for cell by cell.

\begin{theorem}[Resource accounting]\label{thm:resource}
Fix a free skeleton $\mathbf G$ and let $\varepsilon_x=\varepsilon_x(\mathbf G)$ and $\mathcal
R_d=\{x:\varepsilon_x>0\}$ be read against it. For any adaptive protocol and any accepted set,
\begin{equation}\label{eq:resource}
  \magicc(F_A)\;=\;\pacc\,\Dstab(\Ehat)\;\le\;\frac{\sqrt3}{2}\sum_{x\in\mathcal R_d}\varepsilon_x .
\end{equation}
\end{theorem}

\begin{corollary}[Accounting against the free set]\label{cor:resource-free}
Write $\varepsilon^\star_x=\inf_{\widetilde{\mathcal G}_x\in\FInst_x}\|\widetilde{\mathcal M}_x-\widetilde{\mathcal G}_x\|_\diamond$
for the distance of cell $x$ to the free set itself. Then
$\magicc(F_A)\le(\sqrt3/2)\sum_x\varepsilon^\star_x$.
\end{corollary}

\begin{proof}
Write $N$ for the number of cells, which is finite. Fix $\kappa>0$ and choose at each cell a
$\widetilde{\mathcal G}_x\in\FInst_x$ with
$\|\widetilde{\mathcal M}_x-\widetilde{\mathcal G}_x\|_\diamond\le\varepsilon^\star_x+\kappa/N$. This is
possible at every cell by the definition of the infimum, including the cells where it vanishes, so no
closedness of $\FInst_x$ is needed. Applying \cref{thm:resource} to that skeleton gives
$\magicc(F_A)\le(\sqrt3/2)\bigl(\sum_x\varepsilon^\star_x+\kappa\bigr)$, and $\kappa$ was arbitrary.
\end{proof}

\Cref{cor:resource-free} is the form the first bound takes on the protocol alone, with no skeleton in
the statement, and it is available because the first bound needs nothing of the skeleton beyond freeness.
\Cref{thm:suppression} has no such corollary, since its skeleton also carries \cref{cond:h4}, which is
what \cref{def:certificate} is for.

The proof replaces the marked cells by free ones one at a time, uses that the diamond norm does not grow
under pre-composition or post-composition, coarse-grains the transcript to the binary accept-or-reject
channel, and converts norms by $\|Q\|_c\le\sqrt3\|Q\|_\infty$ on one qubit. \Cref{app:thm2} gives it in full. Its content is the converse
reading, that an accepted effect at distance $c$ from the octahedron with acceptance $\pacc$ requires
total resource at least $(2/\sqrt3)\pacc\,c$.

The second bound is the main result, and it uses the class, replacing in exchange the sum over cells by a
sum over connected clusters that must reach the distance scale.

\begin{theorem}[Connected-insertion suppression]\label{thm:suppression}
Let a protocol satisfy \cref{cond:skeleton,cond:degree,cond:h4} at a certified threshold $D$, relative
to a fixed skeleton and a fixed choice of dilations \cref{eq:dilation}, that is, let it carry a
compatible certificate in the sense of \cref{def:certificate}. Put
\begin{equation}\label{eq:lambda}
  \Lambda\;=\;\sum_{\substack{C\subseteq\mathcal R_d\ \mathrm{connected}\\ |C|\ge D}}\ \prod_{x\in C}\eta_x .
\end{equation}
Then, writing $t=e^{\Lambda}-1$,
\begin{equation}\label{eq:suppression}
  \magicc(F_A)\;\le\;\sqrt3\,\min\bigl\{1,\ 2t+3t^2\bigr\} ,
\end{equation}
and in particular $\magicc(F_A)\le 10\sqrt3\,\Lambda$ whenever $\Lambda\le\log 2$. With
$\eta_{\max}=\max_x\eta_x$ and degree at most $\zeta$,
\begin{equation}\label{eq:lambda-count}
  \Lambda\;\le\;m\sum_{k=D}^{m}\bigl(C_\zeta\,\eta_{\max}\bigr)^{k}
\end{equation}
for a constant $C_\zeta$ depending only on $\zeta$.
\end{theorem}

\begin{corollary}[Weak distributed resource]\label{cor:weak}
Let $D\ge\beta d$ be a certified threshold, which \cref{cond:dins} supplies, and read $\Lambda$ against it.
If $m\le d^{\,q}$ for a fixed $q$ and $\eta_{\max}\le a/d$, then
\begin{equation}\label{eq:dlogd}
  \pacc\,\Dstab(\Ehat)\;\le\;e^{-\Omega(d\log d)} .
\end{equation}
Consequently $\pacc\ge d^{-q'}$ forces $\Dstab(\Ehat)\to0$.
\end{corollary}

\begin{proof}
If $\Lambda=0$ then \cref{eq:suppression} gives $\magicc(F_A)=0$ and there is nothing to prove, so assume
$\Lambda>0$. For $d$ large enough that $C_\zeta a/d\le\tfrac12$, the geometric sum in
\cref{eq:lambda-count} is at most twice its first term, so $\Lambda\le2d^{\,q}(C_\zeta a/d)^{\beta d}$
and
\begin{equation*}
  \log\Lambda\;\le\;q\log d+\beta d\log(C_\zeta a)-\beta d\log d+\log 2\;=\;-\beta d\log d+O(d) .
\end{equation*}
Hence $\Lambda\le\log2$ for large $d$, \cref{eq:suppression} applies in its linear form, and
$\magicc(F_A)\le10\sqrt3\,\Lambda\le e^{-\beta d\log d+O(d)}$.
\end{proof}

Two features of \cref{thm:suppression} decide how it should be read. It is stated on the
\emph{amplitudes} $\eta_x$ and not on the diamond strengths $\varepsilon_x$, and
\cref{rem:amplitude} explains why that is the general form. And it is stated on the unnormalised
$\magicc(F_A)=\pacc\,\Dstab(\Ehat)$, which is what makes it immune to a vanishing acceptance rate. A
branch of tiny probability can have its scalar amplified without limit by post-selection, and the
theorem is untouched by that, because \cref{cond:h4} is imposed transcript by transcript before accepted
transcripts are summed. Amplification moves a branch along its free ray and does not rotate it off.

\begin{corollary}[Approximate branch factorisation]\label{cor:approx}
Fix a positive-integer cutoff $D$, which need not be a certified threshold, and read $\Lambda$ against it.
Suppose that for every accepted raw transcript $y$ and every insertion set $S$ whose components all have
fewer than $D$ cells,
\begin{equation}\label{eq:approx-h4}
  \bigl\|(\bra{y}\otimes I_E)V_SP-(J_yP)\otimes\ket{e_{y,S}}\bigr\|\;\le\;\delta_{y,S} ,
\end{equation}
with $J_y$ free and independent of $S$. Put
$\delta_{\mathrm{tot}}=\bigl(\sum_{y\in A}(\sum_{S\ \mathrm{good}}\delta_{y,S})^2\bigr)^{1/2}$ and
$u=(e^{\Lambda}-1)+\delta_{\mathrm{tot}}$. Then
\begin{equation}\label{eq:approx-suppression}
  \magicc(F_A)\;\le\;\sqrt3\,\min\bigl\{1,\ 2u+3u^2\bigr\} .
\end{equation}
\end{corollary}

The error enters exactly where the bad insertion sets do, so an exact-recovery idealisation may be traded
for a stability estimate at the cost of adding $\delta_{\mathrm{tot}}$ to $t$. \Cref{app:thm2} proves it
alongside \cref{thm:suppression}, and \cref{prop:standard} is stated in the exact setting, where
$\delta_{\mathrm{tot}}=0$.

\begin{remark}[Amplitudes, not diamond distances]\label{rem:amplitude}
Continuity of the Stinespring dilation gives $\eta_x=O(\sqrt{\varepsilon_x})$ for \emph{some} choice of
dilations~\cite{KretschmannSchlingemannWerner2008}, and substituting it into \cref{eq:lambda} yields a
bound in the $\varepsilon_x$ alone. The substitution is legitimate exactly when one choice of skeleton,
dilations and record basis achieves the small dilation distance and satisfies \cref{cond:h4} at the same
time, since the environment rotation that realises the continuity bound acts on the $\Delta_x$ and can
mix branches sitting on different rays $J_y$. We therefore state the theorem on the amplitudes and record
the compatibility as a hypothesis where the diamond form is wanted. For cells that are unitary with
one-dimensional environments, which is the case throughout \cref{sec:results-construction}, there is no
gauge freedom and $\eta_x=\Theta(\varepsilon_x)$.
\end{remark}

The exponent of \cref{cor:weak} is $\Dins\log(1/\eta_{\max})$ up to lower-order terms, so the two inputs
enter separately. Code distance supplies the number of cells that must cooperate and the per-cell
strength supplies the logarithm. A constant strength $\eta_{\max}=\eta_0$ with a certified threshold
$\Theta(d)$, polynomially many cells and $C_\zeta\eta_0<1$ gives the upper bound $e^{-\Omega(d)}$ and no
logarithm at all, and the $\log d$ in \cref{eq:dlogd} is a consequence of the
strength shrinking with $d$.

\subsection{Making the branch hypothesis checkable}\label{sec:results-criterion}

\Cref{cond:h4} quantifies over insertion sets and transcripts, so it is stated in a form one does not
check directly. It follows from a property of each recovery slab, which one does check, that each syndrome fibre of small
insertions carries a single recovery coset up to a logical Pauli ambiguity, and that the next free
operation absorbs that ambiguity.

\begin{definition}[Syndrome-fibre absorption]\label{def:absorption}
A recovery slab \emph{absorbs below $D$} when, for every incoming free branch $J_h$ whose range lies in a
single syndrome sector and every syndrome $s$ carrying a nonzero branch, there are a free recovery branch
$R_{h,s}$ that carries syndrome sectors into syndrome sectors, and an abelian group $\mathcal A_{h,s}$ of
logical Paulis
such that every branch operator $E$ produced by an insertion configuration whose components have fewer
than $D$ cells and whose syndrome is $s$ satisfies
\begin{equation}\label{eq:coset}
  \Pi_s E J_h P\;=\;c_E\,\bigl(R_{h,s}L_EJ_hP\bigr)\otimes\ket{f_E}
\end{equation}
for a scalar $c_E$, an environment vector $\ket{f_E}$ and a logical Pauli $L_E\in\mathcal A_{h,s}$; and
the following free branch instrument $\{Q_{h,s,b}\}_b$, whose branches also carry syndrome sectors into
syndrome sectors, absorbs that ambiguity,
\begin{equation}\label{eq:absorb}
  Q_{h,s,b}\,R_{h,s}\,L\,J_hP\;=\;\chi_{h,s,b}(L)\,Q_{h,s,b}\,R_{h,s}\,J_hP
  \qquad\text{for every }L\in\mathcal A_{h,s} ,
\end{equation}
with $\chi_{h,s,b}$ a scalar.
\end{definition}

The definition speaks of a single insertion configuration, and one further hypothesis says how a
configuration with several components is put together from its parts. Call the \emph{dressed support} of
an insertion component $C$ the union over $x\in C$ of the forward light cone of $\Delta_x$ inside the
slab, stopped at the raw-record cut, which is the same object the edges of $\Gfrak_d$ are read from.

\begin{condition}[Componentwise slabs below $D$]\label{cond:components}
Write $R_{h,s}$ and $\mathcal A_{h,s}$ for the witnesses \cref{def:absorption} supplies at that history
and syndrome. Let an insertion configuration have components $C_1,\dots,C_r$ in $\Gfrak_d$, each with
fewer than $D$ cells. Then their dressed supports are pairwise disjoint, the slab syndrome splits as $s=(s_1,\dots,s_r,s_0)$ with $s_j$ a function of $C_j$
alone and $s_0$ independent of the configuration, and the branch operator factorises over the components,
\begin{equation}\label{eq:componentwise}
  \Pi_s\,E\,J_hP\;=\;\Bigl(\prod_{j=1}^{r}c_{j}\Bigr)\,
  \Bigl(R_{h,s}\,\prod_{j=1}^{r}L_{j}\,J_hP\Bigr)\otimes\bigotimes_{j=1}^{r}\ket{f_{j}} ,
\end{equation}
with $c_j\in\mathbb C$, $L_j\in\mathcal A_{h,s}$ and $\ket{f_j}$, all three depending on $C_j$ alone. Any sign a component
contributes belongs in that component's $\ket{f_j}$, which is possible because the eigenvalue of a
stabilizer on a fixed syndrome sector is multiplicative in the stabilizer. The dilation environments
of the marked cells are \emph{terminal}: they are discarded at the raw-record cut, and no register
carried past the cut has been coupled to them, so neither they nor any descendant of them is available
to a later operation.
\end{condition}

\Cref{cond:components} is what the raw-record cut of \cref{cond:skeleton} is for, since components that
are non-adjacent in $\Gfrak_d$ have light cones that do not meet before the cut, and that is disjointness
of dressed supports. The second content of the condition is that the recovery reads and corrects each
component from its own syndrome block, which is a property one reads off a decoder.

The terminal clause is load-bearing and the following shows why. Suppose two disconnected components
leave environment vectors $\ket{b_1}$ and $\ket{b_2}$ at the cut, each depending on its own component, on
a common branch with trivial logical ambiguity. Carry those registers past the cut, apply a controlled-NOT
from the first to the second, and then a physical Pauli $X_q$ on the block controlled by the second. Every
one of those is a free Clifford operation. Configurations whose two components differ in parity now leave
$J_hP$ and $X_qJ_hP$, which are not proportional, so no ray independent of the insertion set survives and
\cref{eq:h4} fails. Letting the environment vector depend on the insertion set does not save this, because
the recombination has already moved that dependence into the system operator. The clause has to reach
descendants and not just the environments themselves, since copying an environment bit into a fresh
ancilla before the cut and recombining the copies afterwards runs the same construction.

\begin{proposition}[Absorption gives branch factorisation]\label{prop:absorption}
If \cref{cond:skeleton} holds, \cref{cond:components} holds at $D$, and every recovery slab absorbs below
$D$, then $D$ is a certified threshold, so \cref{cond:h4} holds at $D$ and $\Dins\ge D$.
\end{proposition}

\begin{proof}
Induct over slabs, carrying two things as the induction hypothesis: that the incoming branch $J_h$ has
its range in a single syndrome sector, and that every insertion configuration below the threshold has by
this point contributed $(J_hP)\otimes\ket{e_{h,S}}$ with $J_h$ free and independent of the configuration.
The base case is the initial branch, which is the identity in the trivial sector, with the empty
configuration contributing $P$. Fix a transcript history $h$ and an insertion configuration all of whose
components have fewer than $D$ cells. By \cref{cond:components} the
configuration contributes through \cref{eq:componentwise}, which has the form of \cref{eq:coset} with
$c_E=\prod_jc_j$, $L_E=\prod_jL_j$ and $\ket{f_E}=\bigotimes_j\ket{f_j}$, so \cref{def:absorption} applies
to it. Applying \cref{eq:absorb} gives
\begin{equation*}
  Q_{h,s,b}\,\Pi_s\,E\,J_hP \;=\; c_E\,\chi_{h,s,b}(L_E)\,\bigl(Q_{h,s,b}R_{h,s}J_hP\bigr)\otimes\ket{f_E} ,
\end{equation*}
so every such insertion term for the transcript extended by $(s,b)$ lies on the single ray
$J_{h,s,b}=Q_{h,s,b}R_{h,s}J_h$, which is free as a composition of free subchannels and does not depend on
the configuration. The scalar $c_E\chi_{h,s,b}(L_E)$ and the environment vector do depend on it, and
\cref{eq:h4} allows that, so the second half of the induction hypothesis is restored at the next slab.
The first half is restored too, since \cref{def:absorption} asks $R_{h,s}$ and every $Q_{h,s,b}$ to carry
syndrome sectors into syndrome sectors, and \cref{cond:skeleton} asks the same of the remaining free
operations of the skeleton, so $J_{h,s,b}P$ again lies in one sector. That ray is the incoming branch of
the next slab, and the induction closes.
\end{proof}

The clauses are what one reads off a circuit. The sector-preservation requirement on $R_{h,s}$ and
$Q_{h,s,b}$ is what keeps the induction of \cref{prop:absorption} inside one syndrome sector at a time. A
Pauli correction satisfies it by mapping each sector onto another sector, which is exactly what a
correction is for, and a measurement of a logical Pauli satisfies it by commuting with the stabilizer
group and so mapping each sector into itself. \Cref{eq:absorb} is written on the incoming branch $J_hP$
and not on $P$, because that is where \cref{prop:absorption} applies it. The first substantive clause
asks that a syndrome fibre of small insertions carry at most one recovery coset, up to a named logical
ambiguity. The second asks that the next free
operation measure that ambiguity in its eigenbasis, or apply it as a correction. A change of Pauli frame
does not do this: an invertible frame update $U$ obeying $UR\bar LP=\chi\,URP$ would force
$\bar LP=\chi P$, which a logical Pauli does not satisfy. The third clause, \cref{cond:components}, asks
that the decoder treat separated components separately, which is where running on one block with a
raw-record cut is used.

\subsection{The threshold on a standard recovery round}\label{sec:results-standard}

\Cref{prop:absorption} certifies a threshold from a circuit, and this subsection carries that out on a
standard syndrome-extraction slab together with a readout chosen for it. The slab is the round a
surface-code memory already performs; the decoder, the post-cut processing and the measurement that ends
the slab are choices, and the statement below fixes the ones it needs.

Fix a CSS code carrying one logical qubit, with $Z$-distance $d_Z$, checks of range $r_0$ and stabilizers
of $+1$ sign, and take the free skeleton to be the identity at every marked slot, $W_x=I$ with a
one-dimensional environment. The slab is that skeleton followed by one exact round of stabilizer
measurement in L\"uders form whose outcomes are written to the raw record, and then free operations
conditioned on that record. Let the marked cells sit on distinct data qubits $q(x)$, one cell per qubit,
in one layer, so that their dilated deviations against $W_x=I$ are $Z$-diagonal,
\begin{equation}\label{eq:zdiag}
  \Delta_x\;=\;I\otimes\ket{u_x}_x\;+\;Z_{q(x)}\otimes\ket{w_x}_x ,
\end{equation}
the marked set and the vectors $\ket{u_x},\ket{w_x}$ being arbitrary. Let $\Gfrak_d$ join $x$ and $x'$
when their qubits lie within range $2r_0$, so that no check meets two components, and on a lattice code of
bounded qubit density this also gives \cref{cond:degree}. Take the incoming branch to have its range in
one syndrome sector, which the identity branch of the first slab does. Assume also that the deviation environments are terminal in the sense of \cref{cond:components}. Finally
fix once and for all a linear section $\sigma\mapsto R(\sigma)$ of the $Z$-error syndrome map on its
image, so that $R(\sigma+\sigma')=R(\sigma)R(\sigma')$, and take the reference recovery of
\cref{def:absorption} to be $R_{h,s}=R(s-s_{\mathrm{in}})$, where $s_{\mathrm{in}}$ is the sector of the
incoming branch. Such a section exists because the syndrome map is $\mathbb F_2$-linear, and its linearity
is what makes the componentwise choice below consistent across configurations.

The instrument that ends the slab is what makes the family worth having, and the choice below is the one
that keeps the accepted effect nontrivial.

\begin{definition}[Split-readout family]\label{def:splitreadout}
A protocol of the above shape is in the \emph{split-readout family} when, after the syndrome round, it
compares the measured syndrome with the incoming one and branches on the difference. If
$s=s_{\mathrm{in}}$ it measures $\Xbar$ natively and accepts the $+1$ outcome. Otherwise it applies the
Pauli correction $R(s-s_{\mathrm{in}})$, measures $\Zbar$ natively, records the outcome and rejects.
\end{definition}

\begin{proposition}[A standard round certifies the $Z$-distance]\label{prop:standard}
Every member of the split-readout family satisfies \cref{cond:components} at $d_Z$, and its slab absorbs
below $d_Z$ in the sense of \cref{def:absorption}. On the accepted fibre $s=s_{\mathrm{in}}$ the ambiguity
group is trivial, and on every other fibre it is $\{I,\Zbar\}$. By \cref{prop:absorption} the protocol
therefore certifies the threshold $D=d_Z$, so $\Dins\ge d_Z$.
\end{proposition}

Three separate things are at work here, and keeping them apart is what makes the statement usable. On a
fibre with nonzero relative syndrome the ambiguity is $\{I,\Zbar\}$ for a reason that has nothing to do
with size, since two $Z$-type terms of equal syndrome differ by an element of the $Z$-type normaliser,
which on a code with one logical qubit is the stabilizer group extended by $\Zbar$, and the $Z$-type
choice of $R_{h,s}$ is what puts the individual $L_E$ in that group and not merely in a coset of it. What
the threshold buys is that a \emph{single} insertion configuration carries one ambiguity element and not
two, which is what \cref{eq:coset} asks for, and that is where the distance enters: two terms of one
configuration differ by a $Z$-type normaliser element supported inside its components, of weight below
$d_Z$ on each, hence a stabilizer. And on the accepted fibre the same weight argument applies to each
surviving $Z_T$ itself, which therefore is a stabilizer, so nothing is left for the $\Xbar$ measurement to
absorb.

The split of \cref{def:splitreadout} is what separates this family from one that measures $\Zbar$
throughout. A native $\Zbar$ measurement absorbs $aI+b\Zbar$ as readily as it absorbs each term, so it
satisfies \cref{eq:absorb} on the accepted branch too, but it also sends every accepted effect into
$\mathcal N$ and the family would carry no magic at all. Measuring $\Xbar$ on the accepted fibre is what
leaves the accepted effect off the $\Zbar$ axis, and the weight argument above is what says the threshold
survives the change.

\begin{proposition}[When the family carries magic]\label{prop:family-magic}
Let a member of the split-readout family have unitary marked cells $e^{i\theta_xZ_{q(x)}}$, and run it on
the first slab. Then the zero-syndrome branch of the full marked layer acts on the code space as
$\alpha I+\beta\Zbar$, where $\beta$ collects the coefficients of those $T\subseteq\mathcal R_d$ whose
$Z_T$ is a logical operator and $\alpha$ those whose $Z_T$ is a stabilizer, and the accepted effect is
\begin{equation}\label{eq:family-effect}
  F_A=(\alpha I+\beta\Zbar)^{\dagger}\,\tfrac12(I+\Xbar)\,(\alpha I+\beta\Zbar) .
\end{equation}
Its magic is positive whenever $\alpha\beta\neq0$ and $|\alpha|\neq|\beta|$, and $\beta\neq0$ requires the
marked set to contain the support of a $Z$-type logical operator, hence at least $d_Z$ cells. Conversely,
if the marked set does contain such a support and the angles range over $(-\pi/2,\pi/2)$, the three
equalities $\alpha=0$, $\beta=0$ and $|\alpha|=|\beta|$ hold only on a closed set of angle vectors of
measure zero, so the accepted effect carries magic at all but exceptional angles.
\end{proposition}

So the family is nontrivial, and what makes it nontrivial can be read off the marked set. A layer whose
support contains no logical $Z$ string has $\beta=0$ and accepts a free effect, while a layer whose
support contains one has $\beta\neq0$ for generic angles and then accepts an effect with magic. The
exception is a matter of the angles and not of the support: on a bare logical string at
$\theta_x=\pi/4$ the two surviving terms have equal modulus, $|\alpha|=|\beta|$, and the accepted effect
returns to the stabilizer boundary. \Cref{cor:weak} bounds the magic of all of them at once.

\begin{remark}[What $d_Z$ is and is not]\label{rem:standard-tight}
\Cref{prop:standard} certifies the value $d_Z$ and therefore gives $\Dins\ge d_Z$. It does not compute
$\Dins$. Suppose the marked set contains the support of a minimum-weight logical $Z$ operator $\gamma$ and
the vectors of \cref{eq:zdiag} are such that both $\prod_{x\in\gamma}\ket{w_x}$ and
$\prod_{x\in\gamma}\ket{u_x}$ contribute. The configuration $S=\gamma$ is one component of $d_Z$ cells
carrying $T=\gamma$ and $T=\varnothing$ at the same syndrome, differing by $\Zbar$, so \cref{eq:coset} has
no single $L_E$ and the coset form of \cref{def:absorption} fails at that size. That is a failure of the
route through \cref{prop:absorption} and not of \cref{cond:h4}, which a $\Zbar$ measurement would still
satisfy by absorbing the sum $aI+b\Zbar$. For the member of \cref{sec:results-construction}, whose
accepted branch ends in an $\Xbar$ measurement, \cref{eq:full-gamma} shows that \cref{cond:h4} itself
fails at $d$, so that member has $\Dins=d$ exactly.
\end{remark}

What \cref{prop:standard,prop:family-magic} settle together is that the class of
\cref{sec:results-class} is not a class of designed protocols. Every single layer of weak $Z$-rotations,
one per qubit, on any set of data qubits and at any angles, is in the split-readout family and certifies
$D=d_Z$. Such a layer accepts an effect carrying magic whenever its support contains a logical $Z$
string and the resulting $\alpha,\beta$ of \cref{prop:family-magic} are nonzero and of different
modulus. \Cref{cor:weak} bounds that magic for all of them at once whenever the angles are $O(1/d)$, the
marked count is polynomial and the degree is bounded. The protocol of \cref{sec:results-construction} is
the member whose marked set is a minimum-weight logical string, and what that member adds is saturation
of the exponent.

\subsection{A protocol that saturates the bound}\label{sec:results-construction}

The class is not empty, and the suppression bound is attained at leading order.

Fix odd $d$ and let $\gamma=\{q_1,\dots,q_d\}$ be a bare non-self-intersecting minimum-weight logical
$\Zbar$ string. Fix a constant $0<\kappa<1$ and put $\theta_d=\kappa/d$.

\begin{definition}[Weak-string protocol]\label{def:weakstring}
Apply $U_j=e^{i\theta_dZ_j}$ at every $q_j\in\gamma$, in one layer. Measure the complete stabilizer
syndrome $s$ and store the raw record. If $s=0$, measure $\Xbar$ natively and accept the $+1$ outcome. If
$s\ne0$, apply the Pauli correction $R(s)$, measure $\Zbar$ natively, record the outcome, and reject.
This is the member of \cref{def:splitreadout} whose marked set is $\gamma$.
\end{definition}

The $\Zbar$ measurement on rejected branches changes no accepted statistic. It is there so that the two
representatives of a syndrome fibre, which differ by the logical $\Zbar$, land on one ray, which is what
\cref{eq:absorb} asks for.

\begin{theorem}[Saturation]\label{thm:construction}
\Cref{def:weakstring} satisfies \cref{cond:one-block,cond:skeleton,cond:degree,cond:h4,cond:dins} with
$m=d$, $\eta_j=2\sin(\theta_d/2)=\Theta(1/d)$ and $\Dins=d$. Its accepted effect is
\begin{equation}\label{eq:constr-effect}
  F_A=\tfrac12\bigl[(\alpha^2+\beta^2)I+(\alpha^2-\beta^2)\Xbar+2\alpha\beta\,\Ybar\bigr],
  \qquad \alpha=\cos^d\theta_d,\quad \beta=\pm\sin^d\theta_d ,
\end{equation}
so that with $t=|\beta|/\alpha=\tan^d\theta_d$,
\begin{equation}\label{eq:constr-magic}
  \pacc=\frac{\alpha^2+\beta^2}{2},\qquad
  \Dstab(\Ehat)=\frac{2t(1-t)}{1+t^2},\qquad
  \magicc(F_A)=\cos^d\theta_d\,\sin^d\theta_d-\sin^{2d}\theta_d .
\end{equation}
Asymptotically $\pacc\to\tfrac12$ and
\begin{equation}\label{eq:constr-asymp}
  -\log\magicc(F_A)\;=\;d\log d+O(d) ,
\end{equation}
which matches the exponent of \cref{eq:lambda-count} at leading order.
\end{theorem}

\Cref{app:localclass} proves it. The step that makes the class hypotheses hold is a property of the bare
string. If a nonempty proper $R\subsetneq\gamma$ had trivial syndrome then $Z_R$ would lie in the
normaliser, and it would be either a logical operator of weight below $d$ or a stabilizer, in which case
$Z_\gamma Z_R$ would represent the same logical operator on $\gamma\setminus R$, again below weight $d$.
Both contradict the code distance, so a syndrome fibre inside $\gamma$ is either a single support or a
complementary pair, and the ambiguity group of \cref{def:absorption} is contained in $\{I,\Zbar\}$.

\subsection{Two protocols that mark the boundary}\label{sec:results-boundary}

The hypotheses are load-bearing, and two constructions outside them show which work is being done by
which.

The first shows that a small resource footprint is not by itself a small resource. Decode the logical
qubit onto a single physical wire by a free Clifford circuit, apply one cell measuring $\Hxy$ on that
wire, and post-select. The accepted effect is the sharp magic projector at $\pacc=1/2$, while the marked
cell occupies one location and its forward causal cone is of constant size. A single insertion already
moves the branch off its free ray, so the largest threshold for which \cref{cond:h4} holds is $\Dins=1$,
and what fails is \cref{cond:dins}. The lesson is that the threshold has to be measured against the
recovery skeleton and not against the geometric size of the resource, because a free circuit can
concentrate the entire logical algebra onto one wire.

The second shows that many weak cells can act coherently. Take $r$ commuting rotations $e^{i\delta Z_j}$
at each site of $\gamma$. Their product is $e^{ir\delta Z_j}$, and in the insertion expansion the $r^d$
temporal choices that select one sine at each site all evaluate to the same logical $\Zbar$, so their
amplitudes add in phase and carry the coefficient $(ir\sin\delta)^d$ in place of $(i\sin\delta)^d$. At
$r=d$ and $\delta=\Theta(1/d)$ this is order one with $m=\Theta(d^2)$ cells of strength $\Theta(1/d)$.
The syndrome-resolved analysis of~\cite{Yoshioka2025} turns the same mechanism into a constant logical
angle at acceptance $\Theta(d^{-1/2})$. What fails here is \cref{cond:degree}. The $r$ co-located
rotations at one site lie within any fixed range of one another, so they form a clique in $\Gfrak_d$ and
the degree is at least $r-1$, which grows with $d$ at $r=d$. The accepted effect of that protocol also
moves by a constant under a single faulted rotation, so $L_F(\delta)=1$ at every scale $\delta$ below
that constant.

\subsection{Faults and protection}\label{sec:results-fault}

\Cref{thm:suppression} is a statement about the accepted effect of a specified protocol, and the
adversary of \cref{sec:prelim} plays no part in it. The fault model enters through a separate bound,
which relates protection to how cheaply the resource can be cut away.

\begin{definition}[Free-replacement distance]\label{def:dfr}
$\dFR$ is the least number of cells whose replacement by resource-erasing free instruments, severing
every quantum and classical path from the remaining marked cells to the accepted-effect functional,
leaves an accepted effect in $\mathcal N$.
\end{definition}

In words, the question is how few cells have to be swapped for free ones before nothing carrying resource
reaches the accepted effect, and the swapped cells are a fault set in the sense of \cref{sec:prelim-fault}.

\begin{proposition}[Replacement bounds protection]\label{prop:replacement}
With $L_F(\delta)$ as in \cref{def:LF}, $L_F(\delta)\le\dFR$ for every $0<\delta\le\magicc(F_A)$.
\end{proposition}

\begin{proof}
A minimising replacement is a fault set of size $\dFR$ in the sense of \cref{sec:prelim-fault}, and it
produces an accepted effect $F^{\text{free}}\in\mathcal N$, so
$\|F_A-F^{\text{free}}\|_c\ge\dist_c(F_A,\mathcal N)=\magicc(F_A)\ge\delta$. That fault set therefore
meets the condition in \cref{eq:LF}, and $L_F(\delta)$ is the minimum over all such sets.
\end{proof}

Read the other way, a protocol protected to $\ell$ at some scale $\delta\le\magicc(F_A)$ has no
free-erasing cut smaller than $\ell$. \Cref{eq:LF} compares unnormalised effects, which is what lets the
proof go through: a replacement that destroys acceptance outright is an allowed competitor, and it is not
excluded by any floor on the faulted acceptance rate, since $0\in\mathcal N$.

%% file: sections/classification.tex
% Section 4 Classification
\section{Where the known constructions sit}\label{sec:classification}

A boundary is most useful as a map. \Cref{prop:exactnative} draws the stabilizer boundary exactly, so a
construction that reaches the magic axis is one that crosses it, and naming the crossing places the
construction. \Cref{thm:suppression} then says what a protocol pays if it crosses by weak distributed
resource on a fixed block, with no added structure. \Cref{cond:one-block} fixes the setting and
\cref{cond:skeleton,cond:degree,cond:h4,cond:dins} are what the bound and its corollary ask for, so
between them they are the five doors through which a construction can leave, and this section names the
first door each known construction uses.

\begin{table}[t]
\centering
\caption{Where surface-code constructions sit. The magic column records whether the accepted effect
carries magic. The next column is the setting of \cref{cond:one-block} and the four after it are the
hypotheses of \cref{cor:weak}, in the order in which they are checked, with a check meaning the condition
holds and a cross meaning it is given up. A dash means
the entry is not evaluated because an earlier one already fails. The last row is outside the fixed-block
setting of \cref{sec:prelim} altogether, since the block it finishes on is not the block the resource
acted on, so no hypothesis is evaluated for it. The table records which condition each construction
leaves and does not rank the constructions by resource. A dagger marks magic at all but exceptional
angles, which \cref{prop:family-magic} gives for unitary cells whose support contains a logical $Z$
string. A double dagger marks the member for which the exponent of \cref{cor:weak} is attained at leading
order.}
\label{tab:classification}
\resizebox{\textwidth}{!}{%
\begin{tabular}{lcccccc}
\toprule
Construction & magic & \cref{cond:one-block} & \cref{cond:skeleton} & \cref{cond:degree} & \cref{cond:h4} & \cref{cond:dins} \\
\midrule
Classical $\Xbar/\Ybar$ mixture & \ding{55} & \ding{51} & \ding{51} & \ding{51} & \ding{51} & \ding{51} \\
Split-readout member, unitary cells on a logical support (\cref{def:splitreadout}) & \ding{51}$^\dagger$ & \ding{51} & \ding{51} & \ding{51} & \ding{51} & \ding{51} \\
Weak-string protocol (\cref{def:weakstring}) & \ding{51}$^\ddagger$ & \ding{51} & \ding{51} & \ding{51} & \ding{51} & \ding{51} \\
Decode-to-one-wire then measure (\cref{sec:results-boundary}) & \ding{51} & \ding{51} & \ding{51} & \ding{51} & \ding{51} & \ding{55} \\
Accumulated weak rotations~\cite{Yoshioka2025} & \ding{51} & \ding{51} & \ding{51} & \ding{55} & --- & --- \\
Transversal continuous rotations~\cite{Huang2025} & \ding{51} & \ding{51} & \ding{51} & \ding{55} & --- & --- \\
Fold-transversal cultivation~\cite{GidneyShuttyJones2024,Claes2025,Sahay2025} & \ding{51} & \ding{51} & \ding{55} & --- & --- & --- \\
Multi-block cultivation~\cite{Vaknin2025} & \ding{51} & \ding{55} & --- & --- & --- & --- \\
Injection and growth~\cite{GidneyShuttyJones2024} & \ding{51} & --- & --- & --- & --- & --- \\
\bottomrule
\end{tabular}}
\end{table}

The reading of \cref{tab:classification} is not that magic and the conditions are incompatible. Every
member of the split-readout family keeps all five, and those members whose support contains a logical $Z$
string carry magic at all but exceptional angles, while those whose support contains none accept a free
effect. \Cref{cor:weak} bounds the magic of all of them by $e^{-\Omega(d\log d)}$ under its
weak-resource hypotheses, and the weak-string protocol is the member for which that exponent is known to
be attained. The reading is that no construction reaching a \emph{constant} amount of magic keeps all
five, and the rest of this section says which one each of those leaves.

\paragraph{Leaving the setting.} Cultivation as originally formulated injects a noisy $\ket{T}$ state on
a small patch and then grows the code distance~\cite{GidneyShuttyJones2024}. The resource acts at one
distance and the state is used at a larger one, so there is no fixed block against which a threshold
could be measured, and the setting of \cref{sec:prelim} does not describe the protocol at all. The
theorem has nothing to say about it and was not built to.

\paragraph{Giving up \cref{cond:one-block}.} A protocol may carry a second protected algebra, and the
multi-block variant of~\cite{Vaknin2025} does, since it projects a pair of low-distance blocks onto a magic eigenstate
before expanding, so the accepted effect is a joint effect on two logical qubits whose distances both
grow, and the one-block hypothesis is the first to go.

\paragraph{Giving up \cref{cond:skeleton}.} A fold-transversal or self-dual
construction~\cite{Claes2025,Sahay2025} changes the block itself. A fold is a weight-$O(d)$ operation, so
it is not a bounded-range cell and the marked cells of such a protocol are not inserted into a skeleton
of bounded-range operations. Placing a fold on the resource scale of \cref{sec:results-free}, in place of
the anyon-permutation scale on which the literature describes it, is an open problem, and
\cref{sec:discussion} returns to it.

\paragraph{Giving up \cref{cond:degree}.} A protocol may place many marked cells at one site, and two
published architectures do. The syndrome-resolved architecture of~\cite{Yoshioka2025} applies many small rotations along a logical
support and resolves the resulting syndrome, and the transversal continuous-rotation protocol
of~\cite{Huang2025} accumulates a target angle over adaptive rounds. In both, the rotations at a single
site lie within any fixed range of one another, so they form a clique in $\Gfrak_d$ whose size grows with
the number of rounds, and the degree bound fails. The mechanism is the phase coherence described in
\cref{sec:results-boundary}, where $r$ commuting rotations at one site contribute $r^d$ insertion
histories that all evaluate to the same logical operator.

\paragraph{Giving up \cref{cond:dins}.} A protocol may satisfy every structural hypothesis and still have
a threshold that does not grow. The decode-to-one-wire construction of \cref{sec:results-boundary} is the
extreme case. A free Clifford circuit concentrates the logical algebra onto one physical wire, one cell
then suffices, and \cref{cond:h4} holds with $\Dins=1$. What fails is the demand that
the threshold grow with the code distance. This is why $\Dins$ must be computed against the recovery
skeleton and not from the geometric size of the resource.

\subsection{Relation to prior no-go theorems}\label{sec:classification-prior}

The results sit in a line of restrictions on logical operations in topological codes, and the boundary
between them is the class of operations each one covers.

Bravyi and K\"onig classify the logical gates implementable by constant-depth geometrically local
\emph{unitary} circuits and show that in two dimensions they are Clifford~\cite{BravyiKoenig2013}, and
Beverland \emph{et al.} extend the classification to topological field
theories~\cite{Beverland2016}. Both concern unitaries. An adaptive, post-selected instrument is not a
unitary, the conjugation and commutator identities their arguments run on are unavailable for it, and the
accepted branch after normalisation is not even a linear map of the input. Measurement is genuinely more
powerful in this setting, and adaptive measurement is known to implement operations that shallow
unitaries cannot~\cite{Aasen2023,WebsterBartlett2018,BravyiKimKlieschKoenig2022}. The results here
therefore do not follow from the unitary classification and do not contradict it.

What replaces the commutator argument is the insertion expansion. It asks how many cells must depart from the free set
before the accepted branch can leave it, and it charges each departure. That is why \cref{thm:suppression} is quantitative where the
unitary results are absolute, and why it survives adaptivity and post-selection.

The relation to the resource theory is the other boundary. Stabilizer operations with post-selection are
free and their closure is known~\cite{Veitch2014,HowardCampbell2017,BravyiGosset2016}, which is
\cref{prop:exactnative}. The operations this paper constrains are not free, so no closure argument
applies to them, and the content is the rate at which their resource is converted into accepted magic on
a block of distance $d$.

\subsection{Relation to the resource theory of dynamical magic}\label{sec:classification-dynamical}

The free set of \cref{def:free} is borrowed from the resource theory of magic for channels, and the
borrowing needs one modification that is worth stating against the literature.

The state-level theory takes the stabilizer polytope as free and measures a state by its distance from
it~\cite{Veitch2014,HowardCampbell2017,HeinrichGross2019}. Lifting this to operations, the free channels
are the completely stabilizer-preserving ones, which map the convex hull of stabilizer states into itself
under every extension, and the resulting theory of dynamical magic supplies channel monotones together
with bounds on conversion, distillation and classical simulation
cost~\cite{SeddonCampbell2019,SaxenaGour2022}. Whether the axiomatically defined completely
stabilizer-preserving channels coincide with the operationally defined stabilizer operations is itself a
settled question, and they do not~\cite{HeimendahlHeinrichGross2022}, which is a reason to be explicit
about which free set a bound is stated against. \Cref{def:free} states ours: an instrument is free when \emph{every} outcome subchannel
is completely stabilizer-preserving.

The modification is the outcome-wise reading, and \cref{sec:results-free} shows why it is forced. A
flagged instrument whose averaged channel is free can have an outcome that performs a magic-axis
measurement, the example being $\mathcal M_\pm(\rho)=\Tr(E_\pm\rho)\ket{0}\bra{0}$ with
$E_\pm=(I\pm\Hxy)/2$, whose sum over outcomes is a constant preparation. A channel-level free set
therefore cannot carry a bound on a post-selected accepted effect, and the free set has to be imposed on
each outcome separately. \Cref{def:free} is the strengthening that does this, and $\mathcal M_\pm$
witnesses that $\FInst_x$ is properly contained in the outcome-forgetting free set: its outcome
subchannel applied to half of an encoded Bell stabilizer state leaves the magic projector $E_\pm^T$ on
the reference, so that outcome is not completely stabilizer-preserving while the averaged channel is.
We emphasise and operationalise this separation for post-selected magic effects, and it is a different
statement from the one~\cite{HeimendahlHeinrichGross2022} establishes.

The quantity being bounded also differs. The channel monotones measure an operation and obey
monotonicity under composition with free operations, which is what makes them useful for conversion
rates. The witness $\Dstab$ of \cref{def:dstab} measures an accepted effect on the encoded qubit, and the
proofs use only that it vanishes on the free set and satisfies the triangle inequality in $\|\cdot\|_c$.
What the results add to the resource theory is the code, since neither of the two bounds follows from a monotonicity argument: \cref{thm:resource} is a first-order
accounting against a fixed skeleton, and \cref{thm:suppression} converts a threshold on how many marked cells
have to cooperate, which \cref{prop:standard} reads off a standard recovery round, into a bound on how
much magic the accepted branch can carry.

Two recent architectures work at the same level and are placed by \cref{tab:classification}. The
syndrome-resolved logical gates of~\cite{Yoshioka2025} and the continuous transversal rotations
of~\cite{Huang2025} both accumulate weak non-Clifford resource on a code block and resolve the resulting
syndrome, and both report a reachable logical angle controlled by the code size. Neither meets
\cref{cond:degree}, for the reason given above, so \cref{cor:weak} does not apply to them and predicts
nothing about them. The resemblance is at the level of the mechanism, where weak cells buy an angle that
shrinks with the block, and turning it into a comparison would need the bound extended to resource that
accumulates across rounds.

%% file: sections/discussion.tex
% Section 5 Discussion
\section{Discussion}\label{sec:discussion}

This section sorts the results by what carries them.

\textbf{The stabilizer boundary.} One statement holds with no further hypothesis. Every accepted effect
of a stabilizer protocol lies in the stabilizer-effect octahedron, for every decoder, every
post-selection rule and every rate of local stochastic Pauli noise, and every point of the octahedron is
reached by a noiseless member of the class, so the reachable set of the class is the octahedron exactly
(\cref{prop:exactnative}). The magic of an
accepted effect equals its acceptance probability times the distance of its normalised effect from that
octahedron (\cref{lem:norm}), and that identity is what lets the later bounds work with an unnormalised
quantity and stay indifferent to a vanishing acceptance rate. The exact reach is a specialisation of the
closure of the free set under stabilizer operations with post-selection, and the surface code contributes
the reading, not the content.

\textbf{First-order accounting.} \Cref{thm:resource} bounds the accepted magic by the summed distance of
the cells from a chosen free skeleton, and \cref{cor:resource-free} by their distance to the free set
itself, with no structural hypothesis at all. Its content is the converse, that constant magic at constant
acceptance costs a constant amount of resource. The quantity it charges is a property of the instrument at each cell, taken outcome
by outcome, which is the only way to charge it. A channel whose outcomes are forgotten can be free while
its outcome statistics measure the magic axis.

\textbf{Suppression.} \Cref{thm:suppression} is the main result, and \cref{cor:weak} turns it into
$e^{-\Omega(d\log d)}$ for polynomially many cells of bounded insertion degree and amplitude $O(1/d)$ at a
threshold linear in the distance. \Cref{prop:absorption} reduces the branch
hypothesis, given the skeleton and the componentwise slabs, to a per-slab syndrome-fibre criterion, \cref{prop:standard} runs that criterion on one exact
round of stabilizer measurement and certifies the threshold $D=d_Z$ for a single layer of weak
$Z$-rotations followed by the split readout of \cref{def:splitreadout}, and \cref{thm:construction}
exhibits a member of that family attaining the exponent at leading order. The hypotheses are therefore met
by a family and not by a single design, \cref{prop:family-magic} gives a necessary condition for a member
to carry magic and generic sufficiency for unitary cells on a logical support, and on the weak-string member
the bound is not loose.

\textbf{What sets the scale.} The exponent is $\Dins\log(1/\eta_{\max})$, and the two factors come from
different places. The threshold fixes how many cells must act together before a branch can leave the free
set, and the per-cell amplitude fixes what each of them costs. The logarithm is a consequence of letting
the amplitude shrink with $d$, and constant-amplitude cells at the same threshold give the weaker upper
bound $e^{-\Omega(d)}$, provided the per-cell amplitude stays below the reciprocal of the degree
constant.
Nothing in the argument uses $\Braid(\mm,\eps)=-1$ or any other feature of the toric-code phase, so the
statement transfers to any stabilizer code whose recovery satisfies the same hypotheses. This is worth
stating plainly, because the anyon picture is how the difficulty of a magic check on a surface-code patch
is usually explained, and it plays no part in what is proved here. What replaces it is the free set of
the resource theory of magic together with the code distance.

\textbf{Where the hypotheses bind.} \Cref{sec:results-boundary} gives two protocols that leave them. In
the first, a free Clifford circuit concentrates the logical algebra onto one physical wire and a single
non-free cell then produces the sharp magic projector at acceptance one half. Its resource has constant
geometric size, which is why the threshold has to be measured against the recovery skeleton and not
against the footprint of the resource. In the second, $r$ commuting rotations at each site of a string
contribute $r^d$ insertion histories evaluating to the same logical operator, and their amplitudes add
in phase. The syndrome-resolved analysis of~\cite{Yoshioka2025} turns that mechanism into a constant
logical angle at acceptance $\Theta(d^{-1/2})$, and the transversal-rotation protocol
of~\cite{Huang2025} reports a reachable logical angle shrinking with code size, which is the behaviour
the exponent predicts for weak cells. Neither is in the class, the first because a single insertion
already moves its branch off the free ray, the second because the co-located rotations at one site form a
clique in $\Gfrak_d$ whose size grows with the number of rounds, so the degree is unbounded. Forming one
connected cluster is not itself a failure, and the weak-string protocol has one.

\textbf{The known cultivation protocols.} Each leaves the hypotheses by a different door, and
\cref{sec:classification} tabulates them. Injection places a resource that is already outside the free
set, folds and code switching change the block, and growth means the distance at which the resource acts
is not the distance at which the state is finally used. The theorem therefore does not forbid these
constructions and was not built to, and what it says is what a protocol pays if it declines all of those
mechanisms and works with weak distributed resource on a fixed block.

\textbf{Approximate factorisation.} \Cref{cor:approx} replaces $t$ by
$t+\delta_{\mathrm{tot}}$, so the bound is stable in the aggregate $\delta_{\mathrm{tot}}$ of
\cref{eq:approx-h4}. That aggregate runs over the good insertion sets of every accepted transcript, and both
of those families are large, so it is the aggregate and not the per-history error that the corollary asks to
be small. An error model carries \cref{thm:suppression} over to noisy syndrome extraction when it holds
each $\delta_{y,S}$ below the target for $\delta_{\mathrm{tot}}$, divided by the number of good sets and
by the square root of the number of accepted transcripts.

\textbf{Protection.} \Cref{thm:suppression} is a statement about the accepted effect of a specified
protocol and the adversary of \cref{sec:prelim} plays no part in its proof. The fault model enters
through \cref{prop:replacement}, which says that a protocol protected to $\ell$ at some scale
$\delta\le\magicc(F_A)$ has no resource-erasing free replacement smaller than $\ell$. The two are
separate results with separate hypotheses, and reading either as evidence for the other would be a
mistake.

\textbf{The two scales.} The framework measures a logical measurement against the free set of the resource
theory and charges it for the distance. Nothing in it is specific to $\Hxy$, which enters only as the
point of the Bloch sphere the witness is evaluated at. It is formulated for one logical qubit, with the
one-qubit octahedron as its free set and three-dimensional norm conversions. Three questions about the
relation between that resource scale and the geometry of the block follow from it. The first is what a
general geometric lower bound on $\Dins$ would look like. \Cref{prop:standard} reads the threshold
$D=d_Z$ off the code distance for a single layer of $Z$-diagonal single-qubit deviations inserted into one
exact syndrome-extraction slab followed by the split readout of \cref{def:splitreadout}, and a bound that
read $\Dins$ off the geometry in general would turn the resource count into a statement about where on the
block the resource has to sit. The second is where a fold sits on the resource scale. A fold is described
in the literature by the anyon automorphism it implements, and its diamond distance from the free set is
a different quantity that nobody has computed. The third is whether the two scales agree at all, that is,
whether an operation's distance from the free set of the resource theory can be read off its action on
the anyon labels. The results here say only that the second does not determine the first, since the two
filters of \cref{eq:twofilter} are syndrome-preserving and their composition has an accepted effect
outside $\mathcal N$.

%% file: appendix/B-thm2.tex
% Appendix B
\section{Proofs of Theorem~\ref{thm:resource} and Theorem~\ref{thm:suppression}}\label{app:thm2}

Both proofs run on the insertion expansion of \cref{eq:insertion}. \Cref{sec:appB-norm} fixes the one-qubit
norms and proves the normalisation identity that turns the magic witness into a distance.
\Cref{sec:appB-resource} proves \cref{thm:resource} by replacing marked cells one at a time.
\Cref{sec:appB-cylinder} sets up the cylinder estimate and the exact-component expansion, and
\cref{sec:appB-assembly} assembles \cref{thm:suppression}.

Write the code-space part of the accepted effect on the logical qubit as
\begin{equation}
  F=P_{\text{code}}\Eacc P_{\text{code}}=p\,(I+\bm v\cdot\bm\sigma),\qquad p=\pacc=\tfrac12\Tr F,
\end{equation}
and let the native stabilizer-effect cone be
\begin{equation}
  \mathcal N=\{\,q(I+\bm w\cdot\bm\sigma):q\ge0,\ \|\bm w\|_1\le1\,\},
\end{equation}
so that $F\in\mathcal N$ exactly when the normalised effect lies in the stabilizer-effect octahedron. Recall
$\Dstab(\Ehat)=\max(0,\|\bm v\|_1-1)$.

\subsection{Normalisation}\label{sec:appB-norm}

We measure distances with the coefficient $\ell_1$ norm $\|aI+\bm x\cdot\bm\sigma\|_c=|a|+\|\bm x\|_1$ (written with
subscript $c$ to distinguish it from the quantum diamond norm) and set $\mathrm{magic}_c(F)=\dist_c(F,\mathcal N)$.

\begin{lemma}[Normalisation]\label{lem:norm}
$\mathrm{magic}_c(F)=\pacc\,\Dstab(\Ehat)$. Consequently, for any decomposition $F=F_{\text{native}}+F_{\text{bad}}$
with $F_{\text{native}}\in\mathcal N$,
\begin{equation}\label{eq:norm-bound}
  \pacc\,\Dstab(\Ehat)\le\|F_{\text{bad}}\|_c .
\end{equation}
\end{lemma}

\begin{proof}
If $\|\bm v\|_1\le1$ then $F\in\mathcal N$ and both sides vanish. Assume $s=\|\bm v\|_1>1$. For $q\ge0$, the closest
point $q\bm w$ with $\|\bm w\|_1\le1$ to $p\bm v$ in $\ell_1$ distance is at distance $\max(ps-q,0)$, so
\begin{equation}
  \mathrm{magic}_c(F)=\inf_{q\ge0}\Bigl(|p-q|+\max(ps-q,0)\Bigr).
\end{equation}
For $q\le p$ the bracket is $p-q+ps-q\ge p(s-1)$, for $p\le q\le ps$ it equals $q-p+ps-q=p(s-1)$, and for $q\ge ps$
it equals $q-p\ge p(s-1)$. The infimum $p(s-1)$ is attained at $q=p$, giving
$\mathrm{magic}_c(F)=p(s-1)=\pacc\Dstab(\Ehat)$. Since $F_{\text{native}}\in\mathcal N$,
$\|F_{\text{bad}}\|_c=\|F-F_{\text{native}}\|_c\ge\mathrm{magic}_c(F)$, which is
\cref{eq:norm-bound}.
\end{proof}

The same computation gives the exact distance of a normalised effect to the octahedron. Writing
$\Ehat=\tfrac12(I+\bm r\cdot\bm\sigma)$ and $\mathcal O=\{\tfrac12(I+\bm s\cdot\bm\sigma):\|\bm s\|_1\le1\}$,
\begin{equation}\label{eq:exact-dist}
  \dist_c(\Ehat,\mathcal O)=\min_{\|\bm s\|_1\le1}\tfrac12\|\bm r-\bm s\|_1=\tfrac12\Dstab(\Ehat),
\end{equation}
the minimum being attained at $\bm s=\bm r/\|\bm r\|_1$ when $\|\bm r\|_1>1$.

The same normalisation holds in trace norm, up to fixed constants, so the choice of norm is immaterial.

\begin{lemma}[Trace-norm normalisation]\label{lem:trace-norm-normalisation}
Let $\mathrm{magic}_1(F)=\dist_1(F,\mathcal N)=\inf_{G\in\mathcal N}\|F-G\|_1$. Then
\begin{equation}
  \frac{2}{1+\sqrt3}\,\pacc\,\Dstab(\Ehat)\le \mathrm{magic}_1(F)\le 2\,\pacc\,\Dstab(\Ehat).
\end{equation}
\end{lemma}

\begin{proof}
We first prove the one-qubit norm comparison. Let $A=aI+\bm x\cdot\bm\sigma$ with $a\in\mathbb R$, $\bm x\in\mathbb R^3$,
and write $\alpha=|a|$, $r=\|\bm x\|_2$, $s=\|\bm x\|_1$. Since $(\bm x\cdot\bm\sigma)^2=r^2I$, the eigenvalues of $A$
are $\lambda_\pm=a\pm r$, so $\|A\|_1=|a+r|+|a-r|=2\max(\alpha,r)$, while $\|A\|_c=\alpha+s$. The upper bound is
immediate, $\|A\|_1=2\max(\alpha,r)\le2(\alpha+s)=2\|A\|_c$. For the lower bound use $s\le\sqrt3\,r$. If $\alpha\ge r$,
$(\alpha+s)/2\alpha\le(\alpha+\sqrt3 r)/2\alpha\le(1+\sqrt3)/2$, and if $\alpha\le r$, then
$(\alpha+s)/2r\le(r+\sqrt3 r)/2r=(1+\sqrt3)/2$. Hence $\|A\|_1\ge\tfrac{2}{1+\sqrt3}\|A\|_c$. Both constants are sharp:
the factor $2$ at $\bm x=0$, and $2/(1+\sqrt3)$ at $|a|=r$, $\bm x=\tfrac{r}{\sqrt3}(\pm1,\pm1,\pm1)$, where
$s=\sqrt3 r$ and $\|A\|_1=2r$. Applying the comparison to the Hermitian differences $F-G$ with $G\in\mathcal N$ and
taking the infimum gives $\tfrac{2}{1+\sqrt3}\mathrm{magic}_c(F)\le\mathrm{magic}_1(F)\le2\,\mathrm{magic}_c(F)$, and
\cref{lem:norm} then yields the claim.
\end{proof}

Two further elementary facts are used below. The eigenvalues of $Q=aI+\bm x\cdot\bm\sigma$ are
$a\pm\|\bm x\|_2$, so $\|Q\|_\infty=|a|+\|\bm x\|_2$, and Cauchy--Schwarz in three dimensions gives
\begin{equation}\label{eq:norm-convert}
  \|Q\|_c\;\le\;|a|+\sqrt3\|\bm x\|_2\;\le\;\sqrt3\bigl(|a|+\|\bm x\|_2\bigr)\;=\;\sqrt3\,\|Q\|_\infty ,
\end{equation}
with equality only at $a=0$ and $\bm x$ along $(1,1,1)$. And for two effects the binary measurement channels
$\mathcal B_F(\rho)=\Tr(F\rho)\ket{0}\bra{0}+\Tr((I-F)\rho)\ket{1}\bra{1}$ satisfy
\begin{equation}\label{eq:binary-diamond}
  \|\mathcal B_F-\mathcal B_G\|_\diamond\;=\;2\|F-G\|_\infty ,
\end{equation}
since the difference sends $\rho$ to $\Tr((F-G)\rho)(\ket{0}\bra{0}-\ket{1}\bra{1})$, whose trace norm is
$2|\Tr((F-G)\rho)|$, and the output is classical so an entangled reference does not help.

\subsection{Proof of Theorem~\ref{thm:resource}}\label{sec:appB-resource}

Order the marked cells causally as $x_1,\dots,x_m$ and let $\mathcal P^{(k)}$ be the protocol whose first
$k$ marked cells are the actual $\widetilde{\mathcal M}_{x_i}$ and whose remaining marked cells are the
skeleton members $\widetilde{\mathcal G}_{x_i}$. Then $\mathcal P^{(m)}$ is the protocol itself and
$\mathcal P^{(0)}$ is free at every marked cell. The diamond norm does not increase under composition
with a fixed channel on either side, so
\begin{equation*}
  \bigl\|\mathcal P^{(k)}-\mathcal P^{(k-1)}\bigr\|_\diamond
  \;\le\;\bigl\|\widetilde{\mathcal M}_{x_k}-\widetilde{\mathcal G}_{x_k}\bigr\|_\diamond
  \;=\;\varepsilon_{x_k} ,
\end{equation*}
and summing gives $\|\mathcal P-\mathcal P^{\text{free}}\|_\diamond\le\sum_x\varepsilon_x$.
Coarse-graining the transcript register to the binary accept-or-reject channel is composition with a
fixed channel, so it does not increase the distance, and \cref{eq:binary-diamond} turns the result into
$\|F_A-F_A^{\text{free}}\|_\infty\le\tfrac12\sum_x\varepsilon_x$. Every accepted branch of the replaced
protocol is a free subchannel, so \cref{lem:free-native} places $F_A^{\text{free}}\in\mathcal N$, and
\begin{equation*}
  \magicc(F_A)\;=\;\dist_c(F_A,\mathcal N)\;\le\;\|F_A-F_A^{\text{free}}\|_c
  \;\le\;\sqrt3\,\|F_A-F_A^{\text{free}}\|_\infty\;\le\;\frac{\sqrt3}{2}\sum_{x}\varepsilon_x
\end{equation*}
by \cref{eq:norm-convert}, which is \cref{eq:resource}.

\subsection{Cylinders and exact components}\label{sec:appB-cylinder}

For disjoint $I,B\subseteq\mathcal R_d$ put
\begin{equation}\label{eq:cylinder}
  V[I;B]\;:=\;\sum_{\substack{S\supseteq I\\ S\cap B=\varnothing}}V_S .
\end{equation}

\begin{lemma}[Cylinder estimate]\label{lem:cylinder}
$\|V[I;B]\|\le\prod_{x\in I}\eta_x$.
\end{lemma}

\begin{proof}
By multilinearity of \cref{eq:insertion} in the marked slots, $V[I;B]$ is the circuit carrying $\Delta_x$
at $x\in I$, carrying $W_x$ at $x\in B$, and carrying the full $V_x=W_x+\Delta_x$ at every other marked
slot. Each $V_x$ and each $W_x$ is a contraction and every unmarked cell is an isometry, so the norm is
at most $\prod_{x\in I}\|\Delta_x\|$.
\end{proof}

The estimate charges nothing for the forced non-insertions on $B$, which is what allows the counting
below to pin down exact components, not mere containment.

For a connected $C\subseteq\mathcal R_d$ write $\partial C=\{x\notin C:\exists z\in C,\ x\sim z\}$ in
$\Gfrak_d$, and let
\begin{equation}\label{eq:exact-comp}
  E_C\;=\;\{S:\ C\subseteq S,\ S\cap\partial C=\varnothing\}
\end{equation}
be the insertion sets for which $C$ is an \emph{exact} connected component. Then $V[E_C]=V[C;\partial C]$
and $\|V[E_C]\|\le w(C):=\prod_{x\in C}\eta_x$. Call a collection $\mathscr C$ of connected sets
\emph{compatible} when its members are pairwise disjoint and no edge of $\Gfrak_d$ joins two of them. If
two members overlap or are adjacent then $E_C\cap E_{C'}=\varnothing$, since one induced subgraph cannot
have both as exact components. For a compatible collection the sets $I_{\mathscr C}=\bigcup_CC$ and
$B_{\mathscr C}=\bigcup_C\partial C$ are disjoint and
$\bigcap_{C\in\mathscr C}E_C=\{S:I_{\mathscr C}\subseteq S,\ S\cap B_{\mathscr C}=\varnothing\}$, so
\cref{lem:cylinder} gives
\begin{equation}\label{eq:compat-bound}
  \Bigl\|V\Bigl[\bigcap_{C\in\mathscr C}E_C\Bigr]\Bigr\|\;\le\;\prod_{C\in\mathscr C}w(C) .
\end{equation}

\begin{lemma}[Exact-component expansion]\label{lem:inclusion}
Let $\mathcal K_D$ be the connected subsets of $\mathcal R_d$ with at least $D$ cells and let
$V_{\mathrm{bad}}=V[\bigcup_{C\in\mathcal K_D}E_C]$. Then, with $\Lambda=\sum_{C\in\mathcal K_D}w(C)$,
\begin{equation}\label{eq:vbad}
  \|V_{\mathrm{bad}}\|\;\le\;\prod_{C\in\mathcal K_D}\bigl(1+w(C)\bigr)-1\;\le\;e^{\Lambda}-1 .
\end{equation}
\end{lemma}

\begin{proof}
The family $\mathcal K_D$ is finite and $\mathcal A\mapsto V[\mathcal A]$ is additive over disjoint
families, so inclusion and exclusion applied coefficientwise to \cref{eq:insertion} gives the exact
identity
\begin{equation*}
  V_{\mathrm{bad}}\;=\;\sum_{\varnothing\ne\mathscr C\subseteq\mathcal K_D}(-1)^{|\mathscr C|+1}
  V\Bigl[\bigcap_{C\in\mathscr C}E_C\Bigr] .
\end{equation*}
Incompatible collections contribute zero. Taking norms, discarding signs and applying
\cref{eq:compat-bound} bounds the right side by $\sum_{\varnothing\ne\mathscr C}\prod_Cw(C)$, which is
$\prod_C(1+w(C))-1$, and $1+u\le e^u$ closes it.
\end{proof}

Conditioning on \emph{exact} components, not on containment, is what makes \cref{lem:inclusion} an
identity. Enumerating one oversized component per bad set and leaving the unexamined cells carrying the
full $V_x$ would cover every bad set without decomposing $V_{\mathrm{bad}}$, since a set with two
oversized components lies in two such cylinders, and restricting to the first of them is a condition on
the unexamined cells that the full $V_x$ does not impose.

\subsection{Proof of Theorem~\ref{thm:suppression}}\label{sec:appB-assembly}

Split $V=V_{\mathrm{good}}+V_{\mathrm{bad}}$ at the certified threshold $D$, so $V_{\mathrm{good}}$ collects the insertion
sets all of whose components are smaller than $D$. Fix a complete raw transcript $y$. By
\cref{cond:h4} every good $S$ contributes $(\bra{y}\otimes I_E)V_SP=(J_yP)\otimes\ket{e_{y,S}}$ with one
$J_y$, so
\begin{equation*}
  (\bra{y}\otimes I_E)V_{\mathrm{good}}P
  \;=\;(J_yP)\otimes\Bigl(\sum_{S\ \mathrm{good}}\ket{e_{y,S}}\Bigr) ,
\end{equation*}
and therefore
\begin{equation}\label{eq:fgood}
  F_{\mathrm{good}}\;=\;\sum_{y\in A}\Bigl\|\sum_{S\ \mathrm{good}}e_{y,S}\Bigr\|^2\,PJ_y^\dagger J_yP .
\end{equation}
Every coefficient is nonnegative, every $PJ_y^\dagger J_yP$ lies in $\mathcal N$ by
\cref{lem:free-native}, and $\mathcal N$ is a convex cone, so $F_{\mathrm{good}}\in\mathcal N$. This is
where post-selection is disposed of. The accepted set enters only as the range of the outer sum in
\cref{eq:fgood}, and restricting it deletes nonnegative terms. A branch of small probability may have its
coefficient amplified without limit, and it stays on the ray $J_y$ throughout.

Let $t=\|V_{\mathrm{bad}}\|$. Since $\|V\|\le1$ we have $\|V_{\mathrm{good}}\|\le1+t$, so with $Q_A$ the
projector onto accepted transcripts,
\begin{equation*}
  \|F-F_{\mathrm{good}}\|_\infty
  \;=\;\bigl\|P\bigl(V_{\mathrm{good}}^\dagger Q_AV_{\mathrm{bad}}+V_{\mathrm{bad}}^\dagger Q_AV_{\mathrm{good}}
  +V_{\mathrm{bad}}^\dagger Q_AV_{\mathrm{bad}}\bigr)P\bigr\|_\infty
  \;\le\;2(1+t)t+t^2\;=\;2t+3t^2 .
\end{equation*}
Since $F_{\mathrm{good}}\in\mathcal N$, \cref{eq:norm-convert} gives $\magicc(F)\le\sqrt3(2t+3t^2)$. The
bound $\magicc(F)\le\|F\|_c\le\sqrt3$ holds always, because $0\in\mathcal N$ and $0\le F\le I$, and the
minimum of the two is \cref{eq:suppression}. \Cref{lem:inclusion} supplies $t\le e^{\Lambda}-1$.

For the linear form let $\Lambda\le\log2$. Then $t=e^{\Lambda}-1\le1$, so $2t+3t^2\le5t$, and
$e^{\Lambda}-1\le\Lambda e^{\Lambda}\le2\Lambda$ on that range, giving
$\magicc(F)\le5\sqrt3\cdot2\Lambda=10\sqrt3\,\Lambda$.

It remains to count. In a graph of degree at most $\zeta$ the number of connected subsets of size $k$
containing a fixed vertex is at most $C_\zeta^{\,k}$ for a constant depending only on $\zeta$, so there
are at most $mC_\zeta^{\,k}$ connected subsets of size $k$ and
\begin{equation*}
  \Lambda\;=\;\sum_{C\in\mathcal K_D}\prod_{x\in C}\eta_x\;\le\;m\sum_{k=D}^{m}\bigl(C_\zeta\eta_{\max}\bigr)^{k} ,\qquad D\in\mathbb N ,
\end{equation*}
which is \cref{eq:lambda-count}. The upper limit $k=m$ is what keeps the sum finite when
$C_\zeta\eta_{\max}\ge1$, in which range \cref{eq:suppression} falls back on its constant branch.

\begin{proof}[Proof of \cref{cor:approx}]
Keep the split $V=V_{\mathrm{good}}+V_{\mathrm{bad}}$ and let $W$ be the operator defined by
$(\bra{y}\otimes I_E)WP=(J_yP)\otimes\ket{e_{y,S}}$ summed over good $S$, so that $W$ is what
\cref{cond:h4} would have produced exactly. Put $\Delta=Q_A(V_{\mathrm{good}}-W)P$. Distinct raw
transcripts are orthogonal in the record register, so $\Delta=\sum_y\ket{y}\otimes\Delta_y$ with
$\|\Delta_y\|\le\sum_{S\ \mathrm{good}}\delta_{y,S}$ by \cref{eq:approx-h4}, and therefore
$\|\Delta\|\le\delta_{\mathrm{tot}}$.

The operator $F_{\mathrm{good}}^{\mathrm{ex}}=PW^\dagger Q_AWP$ is \cref{eq:fgood} verbatim and lies in
$\mathcal N$ by the same convexity argument. Writing $U=\Delta+Q_AV_{\mathrm{bad}}P$ we have
$\|U\|\le\delta_{\mathrm{tot}}+t=u$ and $\|Q_AWP\|\le\|Q_AV_{\mathrm{good}}P\|+\delta_{\mathrm{tot}}\le1+u$,
so
\begin{equation*}
  \|F-F_{\mathrm{good}}^{\mathrm{ex}}\|_\infty
  \;=\;\bigl\|P\bigl(W^\dagger Q_AU+U^\dagger Q_AW+U^\dagger Q_AU\bigr)P\bigr\|_\infty
  \;\le\;2(1+u)u+u^2\;=\;2u+3u^2 .
\end{equation*}
\Cref{eq:norm-convert} and the constant bound $\magicc(F)\le\sqrt3$ then give
\cref{eq:approx-suppression}. At $\delta_{\mathrm{tot}}=0$ this is \cref{eq:suppression}.
\end{proof}

\begin{remark}[Where each hypothesis is used]\label{rem:where-used}
\Cref{cond:h4} is used once, at \cref{eq:fgood}, and it is the only place the accepted set is handled.
\Cref{cond:degree} is used once, in the counting above. \Cref{cond:skeleton} is what stops the light
cones at a raw-record cut so that $\Gfrak_d$ is the graph the counting runs on, and
\Cref{cond:components} is not a hypothesis of this theorem and is used nowhere in this appendix. It
enters upstream, in \cref{prop:absorption}, which is one route to \cref{cond:h4}. \Cref{cond:one-block}
is not a hypothesis either and is used nowhere here; it fixes the setting, as
\cref{sec:results-class} says.
\Cref{cond:dins} enters only in \cref{cor:weak}.
\end{remark}

%% file: appendix/E-example.tex
% Appendix E
\section{A worked example: the distance-3 patch}\label{app:example}

We make the objects of the theorem concrete on the smallest nontrivial patch, the distance-3 rotated surface code with
nine data qubits, which is the setting a first cultivation attempt uses. On that patch we compare a fine stabilizer
transcript, a classical coarse-graining, and a fold-supported check.

\paragraph{The patch and its logical strings.} The nine data qubits sit on a $3\times3$ grid. The stabilizer group is
generated by the weight-4 bulk plaquettes and vertices together with the weight-2 boundary generators, eight
independent generators in all, leaving one logical qubit. The boundaries are rough top and bottom and smooth left and
right. The logical $\Xbar$ is a horizontal weight-3 $X$-string between the smooth boundaries, the logical $\Zbar$ is a
vertical weight-3 $Z$-string between the rough boundaries, and their product $\Ybar=i\Xbar\Zbar$ runs along the
diagonal. Any $X$-type operator commuting with every stabilizer can be multiplied by stabilizers until it is either
the identity or that horizontal line, so the only invariant of a logical representative on this patch is whether it
connects the two smooth boundaries. The three logical operators satisfy $\Xbar\Zbar=-\Zbar\Xbar$, and that single
anticommutation is the algebraic fact the rest of this appendix uses.

\paragraph{Why a stabilizer check fails.} Suppose we try to check $\Hxy$ with stabilizer operations alone. A single
transcript of Pauli and stabilizer measurements can measure $\Xbar$, or measure $\Ybar$, or interleave the two,
but by \cref{lem:collapse} its accepted effect is a stabilizer projector $\Pi_S$ with $S$ abelian. On the $d=3$ patch
one checks directly that $\Xbar$ and $\Ybar$ cannot both lie in a commuting group with a nonzero joint eigenspace:
they are conjugate logical operators, $\Xbar\Ybar=-\Ybar\Xbar$, so any $\Pi_S$ containing representatives of both
annihilates the code space. One can see the collapse in a single line. If the transcript first measures $\Xbar$ with outcome $+$ and then $\Ybar$,
the accepted Kraus operator is $K=P_{\Ybar,+}P_{\Xbar,+}$, and its effect is
\begin{equation}
  K^\dagger K=P_{\Xbar,+}P_{\Ybar,+}P_{\Xbar,+}
  =P_{\Xbar,+}\,\frac{I+\Ybar}{2}\,P_{\Xbar,+}
  =\tfrac12\,P_{\Xbar,+},
\end{equation}
using $P_{\Xbar,+}\Ybar P_{\Xbar,+}=0$ since $\Ybar$ anticommutes with $\Xbar$. The second measurement adds nothing to
the effect, which is proportional to the single stabilizer projector $P_{\Xbar,+}$ and so has its Bloch vector on the
single axis $\Xbar$. A simple stabilizer family can coarse-grain by running the $\Xbar$ measurement with probability $\tfrac12$ and the
$\Ybar$ measurement with probability $\tfrac12$ and accepting the $+$ outcome. The
accepted effect is
\begin{equation}
  \Eacc=\tfrac12 P_{\Xbar,+}+\tfrac12 P_{\Ybar,+}=\tfrac12 I+\tfrac14(\Xbar+\Ybar),
\end{equation}
with normalised Bloch vector $\bm v=(\tfrac12,\tfrac12,0)$, so
\begin{equation}
  \|\bm v\|_1=\tfrac12+\tfrac12=1,\qquad \Dstab(\Ehat)=\max(0,\|\bm v\|_1-1)=0 .
\end{equation}
The effect sits on the octahedron boundary. It is not the coherent check but a coin flip between two Pauli
measurements, and it retains which one was performed. This is the classical-mixture loophole closed by the magic
witness, made explicit on the patch. The sharp check would instead need $\bm v=(1/\sqrt2,1/\sqrt2,0)$ with
$\|\bm v\|_1=\sqrt2$, a point the coin flip can never reach because classical mixing only shrinks the $\ell_1$ norm.

\paragraph{What the fold supplies.} A fold-transversal Hadamard folds the patch along its diagonal and applies a
transversal Hadamard, implementing the boundary automorphism $\ee\leftrightarrow\mm$~\cite{Sahay2025,KobayashiZhu2023}.
On the folded lattice the two boundary types are identified, so the code becomes self-dual, and the transversal
Hadamard acts on the logical qubit as the logical Hadamard $\overline H$, exchanging $\Xbar$ and $\Zbar$. At the anyon
level this is the automorphism of $D(\Zt)$ that swaps $\ee$ and $\mm$ and fixes $\eps$.

A single-qubit Clifford conjugation sends a Pauli axis to a Pauli axis, and $\Hxy=(\Xbar+\Ybar)/\sqrt2$ is not a Pauli
operator, so no Clifford brings $\Hxy$ to a Pauli axis and no ordinary boundary readout of the folded code measures
it. What the self-dual structure supplies is a measurement of a \emph{non-Pauli Hermitian logical Clifford
observable}, the logical Hadamard axis $H_{XZ}=(\Xbar+\Zbar)/\sqrt2$, as a fold-supported logical observable. A
single-qubit logical Clifford then relabels the Pauli pair $(\Xbar,\Zbar)$ to $(\Xbar,\Ybar)$, carrying $H_{XZ}$ to
$\Hxy$, which is a change of logical frame and not a second physical resource. Cultivation protocols realise the check
at constant acceptance this way~\cite{GidneyShuttyJones2024,Claes2025,Sahay2025,Vaknin2025}. What the fold costs on the
scale of \cref{sec:results-free} is a separate question from what it does to the anyon labels, and this appendix does
not answer it.

On the $d=3$ patch the fold is a concrete weight-$O(d)$ operation, not a local gate. That is already enough to put it outside \cref{cond:skeleton}, which asks the marked cells to be inserted into a
skeleton of bounded-range operations. The weak-string protocol of \cref{def:weakstring} sits at the other end of the
same patch. Its three rotations at $\theta_3=\kappa/3$ have $\eta_j=\Theta(1/3)$ and $\Dins=3$, and
\cref{eq:constr-magic} gives $\magicc(F_A)=\cos^3\theta_3\sin^3\theta_3-\sin^6\theta_3$ with acceptance close to one
half, which at $\kappa=\tfrac12$ is $4.36\times10^{-3}$.

\paragraph{The general lesson from the small case.} The distance-3 patch is small enough to hold in one's head, and it
already shows the whole argument. The magic axis is a superposition of two logical operators that anticommute, a
single stabilizer transcript can carry only one of them because a commuting group cannot contain an anticommuting
pair, and a classical mixture of the two collapses onto the stabilizer boundary because mixing shrinks the $\ell_1$
norm. The stabilizer transcript and the classical mixture analysed here do not reach the sharp check. The fold reaches
it on the small patch, and it does so with an operation whose weight grows with the patch. Nothing in this picture
used $d=3$ except the convenience of drawing it. At larger distance the strings are longer, the stabilizer group is
bigger, and the fold is a wider operation, but the algebra is the same.
What the small case does not show is the rate, since with $d=3$ every exponent is a constant. That is
what \cref{thm:suppression} supplies, and \cref{app:localclass} is where the $d$-dependence is done.

%% file: appendix/F-localclass.tex
\section{The stabilizer class and its exact reach}\label{app:localclass}

This appendix proves \cref{prop:exactnative}. For protocols whose branches are stabilizer operations the accepted
effect can be computed exactly, and the answer is an equality: the reachable set is the native cone
$\mathcal N$ itself. Nothing from \cref{app:thm2} is used here, and neither is protection or any hypothesis on the noise
rate. \Cref{app:localclass-construction} then proves \cref{thm:construction}, the weak-string protocol
that lies inside the class of \cref{sec:results-class} and saturates \cref{thm:suppression}.

\subsection{Projector collapse}\label{app:localclass-collapse}

The following lemma is the one structural fact about a stabilizer branch that the rest of the appendix uses. A
branch is not literally a product of signed Pauli projectors, and the reduction to one is as follows. Fix a transcript
and a fault configuration, and work on the joint space of the patch and its ancillas. Preparing a product stabilizer
ancilla is a projector onto a stabilizer state, each measurement outcome contributes a signed Pauli projector, each
Clifford or Pauli operation contributes a Clifford unitary, and the branch operator is the ordered product of these.
Move every Clifford factor to the left using $\Pi C=C\,(C^\dagger\Pi C)$, which is legitimate because a Clifford
conjugate of a signed Pauli projector is again one. The branch becomes $K=C\,\Pi_k\cdots\Pi_1$ with $C$ Clifford, and
since $C^\dagger C=I$,
\begin{equation}\label{eq:branch-reduction}
  K^\dagger K\;=\;K'^\dagger K',\qquad K'=\Pi_k\cdots\Pi_1 ,
\end{equation}
so the lemma applied to $K'$ determines the accepted effect of the branch.

\begin{lemma}[Projector collapse]\label{lem:collapse}
Let $K_\tau=\Pi_{k}\cdots\Pi_{1}$ be a product of signed Pauli projectors. Then $E_\tau=c_\tau\,\Pi_S$ for some constant
$c_\tau\ge0$ and some projector $\Pi_S=\lvert S\rvert^{-1}\sum_{g\in S}g$ onto the joint $+1$ eigenspace of a commuting
group $S$ of signed Pauli operators with $-I\notin S$. If $c_\tau>0$, the set of Pauli operators appearing in
$E_\tau$ with a nonzero coefficient is exactly $S$, an abelian group.
\end{lemma}

\begin{proof}
Reduce the palindrome $E_\tau=\Pi_1\cdots\Pi_{k-1}\Pi_k\Pi_{k-1}\cdots\Pi_1$ from the centre outward, keeping the
invariant that the central block is a nonnegative multiple of a stabilizer projector $\Pi_S$ onto a commuting group
$S$ with $-I\notin S$. The innermost factor $\Pi_k=\Pi_{sQ}$ is such a projector, with $S=\langle sQ\rangle$. Assume
the block equals $c\,\Pi_S$ and conjugate by the next projector $\Pi_{tR}=(I+tR)/2$. If $R$ commutes with every
element of $S$, then $\Pi_{tR}\Pi_S\Pi_{tR}=\Pi_S\Pi_{tR}$, which is either $0$ (if $-tR\in S$) or the projector onto
$\langle S,tR\rangle$. If $R$ anticommutes with some element of $S$, split $S=S_0\sqcup S_1$ into the part
$S_0=S\cap C(R)$ commuting with $R$ and its coset $S_1=hS_0$ for any fixed $h\in S$ anticommuting with $R$. Write
$\Pi_S=\lvert S\rvert^{-1}\sum_{g\in S}g=\lvert S\rvert^{-1}\bigl(\sum_{g\in S_0}g+\sum_{g\in S_1}g\bigr)$. Since $R$
anticommutes with every $g\in S_1$, conjugation by $\Pi_{tR}=(I+tR)/2$ kills the $S_1$ sum, because for such $g$
\begin{equation}
  \Pi_{tR}\,g\,\Pi_{tR}=\tfrac14(I+tR)g(I+tR)=\tfrac14\,g(I-tR)(I+tR)=0,
\end{equation}
using $Rg=-gR$ and $R^2=I$, while for $g\in S_0$, which commutes with $R$, $\Pi_{tR}g\Pi_{tR}=g\Pi_{tR}$. Hence
\begin{equation}\label{eq:collapse-anticomm}
  \Pi_{tR}\,\Pi_S\,\Pi_{tR}=\frac{1}{\lvert S\rvert}\sum_{g\in S_0}g\,\Pi_{tR}
  =\frac{\lvert S_0\rvert}{\lvert S\rvert}\,\Pi_{S_0}\Pi_{tR}
  =\tfrac12\,\Pi_{\langle S_0,tR\rangle},
\end{equation}
since $\lvert S_0\rvert=\tfrac12\lvert S\rvert$ and $\Pi_{S_0}\Pi_{tR}=\Pi_{\langle S_0,tR\rangle}$ when $tR$ commutes
with $S_0$ and is consistent with it. In both cases the block stays a nonnegative multiple of a projector onto a
commuting group with $-I$ excluded.
Iterating outward proves $E_\tau=c_\tau\Pi_S$. If $c_\tau=0$ the transcript is contradictory and its accepted effect vanishes,
so it contributes nothing to \cref{eq:accepted-effect} and we may assume $c_\tau>0$. The Pauli operators appearing in
$\Pi_S=\lvert S\rvert^{-1}\sum_{g\in S}g$ are then exactly $S$, an abelian group.
\end{proof}

\Cref{lem:collapse} holds for one fixed transcript, and it fails for a coarse-grained sum, where
$\tfrac12P_{\Xbar,+}+\tfrac12P_{\Ybar,+}$ has both $\Xbar$ and $\Ybar$ in its support. The coarse-grained family is the
domain of \cref{thm:suppression}.

\subsection{The local-noise stabilizer setting}\label{app:localclass-setting}

Fix constants $R,\Delta,M,\lambda<\infty$, independent of $d$. Consider a protocol whose spacetime cell graph
$\mathcal G_d$ has degree at most $\Delta$ and range at most $R$, with fixed rough and smooth boundaries, built from
product stabilizer ancillas, local Clifford and Pauli operations, native Pauli--Wilson measurements, and classical
feed-forward. This is the local stabilizer skeleton of a check-and-grow protocol, and every one of its branches is a
stabilizer operation, which is the only property the proof uses. Noise is local stochastic Pauli noise at rate $p$: after Pauli-frame propagation through the ideal stabilizer
circuit, the set $C(\omega)\subset V(\mathcal G_d)$ of faulted cells obeys $\Pr_p[\,S\subset C(\omega)\,]\le(\lambda
Mp)^{|S|}$ for every finite $S$, with $M$ bounding the fault labels per cell. Finite-range correlated Pauli noise with
the same domination is allowed. The proof below uses none of this beyond the fact that the noise is a convex mixture
of Pauli conjugations, and it holds at every rate $p$.

The class is closed under classical mixing, which we record separately because the converse of
\cref{prop:exactnative} uses it and because feed-forward on its own supplies only dyadic weights. Given members
$A^{(1)},\dots,A^{(k)}$ of the class and weights $w_j\ge0$ summing to one, the protocol that draws $j$ from $w$ and
runs $A^{(j)}$ is again a member, with accepted effect
\begin{equation}\label{eq:classmix}
  F \;=\; \sum_{j=1}^{k} w_j\,F^{(j)} .
\end{equation}
The draw is classical data and it touches no qubit of the patch, so every branch of the mixture is a stabilizer
operation exactly as before.

\subsection{Every branch is a native effect}\label{app:localclass-gains}

Local stochastic Pauli noise is a convex mixture of Pauli conjugations, so writing $\omega$ for a fault configuration,
one Pauli per cell, the noisy circuit is
\begin{equation}\label{eq:stoch-mixture} \mathcal E \;=\; \sum_\omega \Pr\nolimits_p(\omega)\,\mathcal E_\omega,
\qquad \Pr\nolimits_p(\omega)\ge0,\qquad \sum_\omega\Pr\nolimits_p(\omega)=1, \end{equation} where $\mathcal E_\omega$
is the ideal stabilizer circuit with the Paulis of $\omega$ inserted. For a fixed $\omega$ and a fixed transcript
$\tau$ the branch is a stabilizer operation, so $K_{\tau,\omega}$ is a product of Pauli operators and signed Pauli
projectors and \cref{lem:collapse} applies to it directly. We use only the effect-level consequence
$K_{\tau,\omega}^\dagger K_{\tau,\omega}=c_{\tau,\omega}\Pi_S$, and never a normal form for $K_{\tau,\omega}$ itself,
which a sequence of Pauli measurements need not have. Already on the logical qubit the branch
$\Pi_{+\Zbar}\Pi_{+\Xbar}$ has $K^\dagger K=\tfrac12\Pi_{+\Xbar}$, whose scalar rules out writing $K$ as a Pauli times
a projector. The atoms of the expansion are the pairs $(\tau,\omega)$, contributing
\begin{equation}\label{eq:stoch-atom}
  B_{\tau,\omega} \;=\;\Pr\nolimits_p(\omega)\;P_{\text{code}}\,K_{\tau,\omega}^\dagger
  K_{\tau,\omega}\,P_{\text{code}} ,
\end{equation}
which is additive over $\tau$ and over $\omega$, so the accepted effect is the sum of the atoms. The weights
$\Pr_p(\omega)$ are nonnegative, and together with the branch structure that is the whole input to the next lemma.

\begin{lemma}[Branch form]\label{lem:branch-native}
Let $K$ be a product of Pauli operators and signed Pauli projectors, so that $K^\dagger K=c_0\Pi_S$ by
\cref{lem:collapse}, with $c_0\ge0$ and $S$ a commuting group of signed Paulis with $-I\notin S$. Then
\begin{equation*}
  P_{\mathrm{code}}\,K^\dagger K\,P_{\mathrm{code}}\;=\;c\,\Pi,\qquad c\ge0,
\end{equation*}
where $\Pi$ is either $0$ or, restricted to the code space, one of $I$ and $\tfrac12(I+s\bar P)$ with $s=\pm1$ and
$\bar P$ a logical Pauli operator. In particular it lies in $\mathcal N$.
\end{lemma}

\begin{proof}
The code projector $P_{\mathrm{code}}=\Pi_{\mathcal G}$ is a product of signed Pauli projectors, one per stabilizer
generator, and so is $\Pi_S$, so $\Pi_SP_{\mathrm{code}}$ is a product of signed Pauli projectors. Its palindrome is
$P_{\mathrm{code}}\Pi_S\Pi_SP_{\mathrm{code}}$, which equals $P_{\mathrm{code}}\Pi_SP_{\mathrm{code}}$ because
$\Pi_S^2=\Pi_S$, and the two families of projectors need not commute with each other for \cref{lem:collapse} to
apply. That lemma gives
\begin{equation*}
  P_{\mathrm{code}}\,\Pi_S\,P_{\mathrm{code}}\;=\;c_1\,\Pi_{S'},\qquad c_1\ge0,
\end{equation*}
with $S'$ a commuting group of signed Paulis and $-I\notin S'$. Two cases follow.

If $c_1=0$ the operator is $0$, which lies in $\mathcal N$, and the lemma holds with $c=0$. This is the case in which
$S$ is inconsistent with the code.

If $c_1>0$ then $\operatorname{ran}\Pi_{S'}\subseteq\operatorname{ran}P_{\mathrm{code}}$, so every $g\in\mathcal G$
fixes every vector of $\operatorname{ran}\Pi_{S'}$, giving $g\Pi_{S'}=\Pi_{S'}$ and hence $g\in S'$, because the
pointwise Pauli stabilizer of a nonzero stabilizer projector is its own signed stabilizer group. So
$\mathcal G\subseteq S'$, and since $S'$ is a commuting group containing $\mathcal G$ we also have
$S'\subseteq N(\mathcal G)$. The image of $S'$ in $N(\mathcal G)/\mathcal G$ is then an \emph{isotropic} subgroup of
the one-qubit logical symplectic space $\mathbb F_2^{\,2}$, meaning the logical symplectic form vanishes on it, so its
dimension is $0$ or $1$. Correspondingly the restriction of $\Pi_{S'}$ to the code space is $I$ or
$\tfrac12(I+s\bar P)$ for a single logical axis $\bar P$ and a sign $s=\pm1$. Both lie in $\mathcal N$, the first in
its interior and the second on an extreme ray, and $c=c_0c_1\ge0$ keeps the multiple inside the cone. Isotropy is what
carries this step, since the logical Pauli group modulo phases is abelian as an abstract group and commutation there
would say nothing.
\end{proof}

\Cref{lem:branch-native} never asks where in the patch the faults sit, and that is worth noticing. A branch carrying
faults is still a stabilizer branch, its accepted effect still collapses to a projector, and a projector on one
logical qubit has its Bloch vector along a single axis. Faults inside a stabilizer branch move the axis, and moving
the axis is something a native Pauli measurement already does. Cone membership therefore follows from the projector
structure alone, with no input from the geometry of the patch or from where the faults are.

\begin{proof}[Proof of \cref{prop:exactnative}]
For the forward direction, fix an accepted set $A$ with $\pacc>0$. By \cref{eq:stoch-mixture,eq:stoch-atom} the
accepted effect is
\begin{equation*}
  F_A\;=\;\sum_{\tau\in A}\ \sum_{\omega}\ \Pr\nolimits_p(\omega)\;
  P_{\mathrm{code}}\,K_{\tau,\omega}^\dagger K_{\tau,\omega}\,P_{\mathrm{code}} ,
\end{equation*}
a combination with nonnegative coefficients of the operators of \cref{lem:branch-native}. Each term lies in $\mathcal
N$ and $\mathcal N$ is a convex cone, so $F_A\in\mathcal N$, and normalising by $\pacc>0$ keeps the Bloch vector
inside the octahedron, which is $\Dstab(\Ehat)=0$. The decoder enters only through the choice of $A$, and
restricting the outer sum to $A$ deletes terms from a nonnegative combination, which a cone allows.

For the converse, let $\Ehat=\tfrac12(I+\bm v\cdot\bm\sigma)$ with $t:=\|\bm v\|_1\le1$, and exhibit a noiseless
member of the class, which is the case $p=0$ of the noise model. If $t=0$ the protocol measures $\Zbar$ and accepts
both outcomes, giving $F=I$ and normalised effect $I/2$. If $t>0$, the protocol draws a classical index by
\cref{eq:classmix} and then runs one of two branches. With probability $q=2t/(1+t)$, which lies in $[0,1]$ because
$t\le1$, it selects an axis $\bar P\in\{\Xbar,\Ybar,\Zbar\}$ with probability $|v_P|/t$, measures that logical Pauli
by a native Pauli--Wilson measurement, and accepts the outcome $\operatorname{sign}(v_P)$, the sign being arbitrary
when $v_P=0$. With probability $1-q$ it measures $\Zbar$ and accepts both outcomes. The accepted effect is
\begin{equation*}
  F\;=\;q\sum_{P}\frac{|v_P|}{t}\cdot\frac{I+\operatorname{sign}(v_P)\bar P}{2}\;+\;(1-q)\,I
  \;=\;\frac{1}{1+t}\bigl(I+\bm v\cdot\bm\sigma\bigr) ,
\end{equation*}
using $1-q/2=q/(2t)=1/(1+t)$, so the normalised effect is $\Ehat$. It is a valid effect, since
$\lambda_{\max}(F)=(1+\|\bm v\|_2)/(1+t)\le1$. Its acceptance probability, which \cref{sec:prelim-effects} fixes as
$\pacc=\tfrac12\Tr F$, the value on the maximally mixed logical input, is $1/(1+t)\ge1/2$ uniformly in $d$. Every cell
is a native Pauli--Wilson measurement and the index is classical data, so every branch is a stabilizer operation and
stabilizer, and every point of the octahedron is attained.
\end{proof}

The two directions meet exactly. \Cref{lem:branch-native} says the branches available to a stabilizer protocol are
nonnegative multiples of $I$ and of the native projectors $\tfrac12(I\pm\bar P)$, and the converse builds an arbitrary
octahedron point out of those same generators. The reachable set is the cone they generate, which is $\mathcal N$.

The two halves treat the noise differently, and the statement is arranged so that they meet. The forward inclusion
holds member by member and at every rate $p$, since \cref{lem:branch-native} holds for every fault configuration. The
converse exhibits the noiseless member $p=0$, and it has to, because at a fixed $p>0$ an undetectable logical fault
already pulls the accepted effect off the extreme rays. Writing
$\mathcal D_r(\rho)=(1-r)\rho+r\Zbar\rho\Zbar$ for such a fault, carried by a weight-$d$ physical representative of
$\Zbar$ at rate $r\le p^{\,d}$, an accepted effect $E_0=aI+x\Xbar+y\Ybar+z\Zbar\in\mathcal N$ is seen as
$aI+(1-2r)(x\Xbar+y\Ybar)+z\Zbar$, and matching the vertex $\tfrac12(I+\Xbar)$ would need $|x|=a/(1-2r)>a$, which
$\mathcal N$ forbids. The reachable set of the class is therefore the octahedron, with the boundary supplied by its
noiseless members.

\subsection{The threshold on a standard recovery round}\label{app:standard}

\begin{proof}[Proof of \cref{prop:standard}]
Write $S$ for the marked cells of an insertion configuration and $C_1,\dots,C_r$ for its components in
$\Gfrak_d$. The skeleton is $W_x=I$ with a one-dimensional environment, so $V_S$ carries $\Delta_x$ at the
cells of $S$ and the identity with no environment factor elsewhere. Expanding \cref{eq:zdiag} at every
cell of $S$, the configuration contributes
\begin{equation}\label{eq:standard-expand}
  E\;=\;\sum_{T\subseteq S}Z_T\otimes\ket{\alpha_T},\qquad
  Z_T=\prod_{x\in T}Z_{q(x)},\qquad
  \ket{\alpha_T}=\bigotimes_{x\in T}\ket{w_x}\bigotimes_{x\in S\setminus T}\ket{u_x} .
\end{equation}
The one cell per qubit hypothesis makes the $q(x)$ for $x\in T$ distinct, so $Z_T$ has weight $|T|$.

Fix a syndrome $s$. The incoming branch has its range in a single syndrome sector $s_{\mathrm{in}}$, so
$Z_TJ_hP$ lies in the sector $s_{\mathrm{in}}+\operatorname{syn}(Z_T)$ and the L\"uders projector $\Pi_s$
keeps exactly the $T$ with $\operatorname{syn}(Z_T)=s-s_{\mathrm{in}}$. This is the step that needs the
incoming branch to be syndrome-definite: for a branch whose range met several sectors, $\Pi_s$ would not
select on $\operatorname{syn}(Z_T)$ alone. Every measured syndrome carrying a nonzero branch therefore
lies in the affine coset $s_{\mathrm{in}}+\operatorname{im}(\operatorname{syn}_Z)$, and those are the
syndromes \cref{def:absorption} quantifies over. Which of them actually occur depends on the
configuration and on the vectors of \cref{eq:zdiag}, and the argument needs only the containment.

\emph{Within one configuration.} Let $T$ and $T'$ both survive, and suppose every component has fewer
than $d_Z$ cells. Then $Z_{T\triangle T'}$ commutes with every check and so lies in the $Z$-type
normaliser. No check meets two components, because a check of range $r_0$ lies in a ball of that radius
and any two of its qubits are therefore within $2r_0$, which would have joined them in $\Gfrak_d$. The
syndrome of $Z_{T\triangle T'}$ consequently splits over the components and vanishes on each. Each
restriction $Z_{(T\triangle T')\cap C_j}$ is a $Z$-type normaliser element of weight at most $|C_j|<d_Z$,
and $d_Z$ is the least weight of a $Z$-type logical operator, so it is a stabilizer. Write
$g=Z_{T\triangle T'}$ for their product, a stabilizer of $+1$ sign. It commutes with $Z_{T'}$, and on the
sector $s_{\mathrm{in}}$ it is a scalar $\epsilon_g(s_{\mathrm{in}})\in\{\pm1\}$, since a stabilizer
is $+1$ on the code space and picks up the commutation sign of the sector's Pauli representative. Hence
$Z_TJ_hP=\epsilon_g(s_{\mathrm{in}})\,Z_{T'}J_hP$, and the signs are the only bookkeeping the step leaves.

Every surviving term therefore equals one fixed $Z_{T_0}$ up to a sign, and with
$L_E=R_{h,s}^{\dagger}Z_{T_0}$,
\begin{equation}\label{eq:standard-coset}
  \Pi_s\,E\,J_hP\;=\;\bigl(R_{h,s}\,L_E\,J_hP\bigr)\otimes\ket{f_E},\qquad
  \ket{f_E}=\sum_{T\ \mathrm{surviving}}\epsilon_T\ket{\alpha_T} ,
\end{equation}
with $\epsilon_T$ the sign relating $Z_T$ to $Z_{T_0}$ on that sector. This is \cref{eq:coset} with
$c_E=1$, the signs having been absorbed into the environment vector.

\emph{The ambiguity group.} Let $T_0$ and $T_0'$ be the representatives of two configurations of the same
syndrome. Then $Z_{T_0\triangle T_0'}$ again lies in the $Z$-type normaliser, and on a code with one
logical qubit that normaliser is the $Z$-type stabilizer group extended by $\Zbar$, so the two
representatives differ by $I$ or $\Zbar$ modulo stabilizers. No size hypothesis enters here. Because
$R_{h,s}=R(s-s_{\mathrm{in}})$ is a $Z$-type Pauli whose syndrome is exactly that of the surviving $Z_T$,
each $L_E=R_{h,s}^{\dagger}Z_{T_0}$ is itself a $Z$-type normaliser element of trivial syndrome, hence
lies in $\{I,\Zbar\}$ modulo stabilizers and not merely in a coset of it. That group is abelian, as
\cref{def:absorption} requires.

\emph{\Cref{cond:components} at $d_Z$.} Take a configuration all of whose components have fewer than
$d_Z$ cells, which is the regime \cref{cond:components} speaks about. The dressed supports are disjoint by
the same check-separation argument.
Because no check meets two components, the syndrome constraint $\operatorname{syn}(Z_T)=s-s_{\mathrm{in}}$
splits into one constraint per component, involving only $T\cap C_j$, and the residual block $s_0$ is
independent of the configuration because the skeleton is the identity at every marked slot. Writing $s_j$
for the block of $s-s_{\mathrm{in}}$ carried by the checks that meet $C_j$, so that
$s-s_{\mathrm{in}}=\sum_js_j$, linearity of the section gives $R(s-s_{\mathrm{in}})=\prod_jR(s_j)$, and the
single fixed recovery therefore factorises over the components of every configuration at once.
\Cref{eq:standard-expand} factorises over components by construction, which is \cref{eq:componentwise}
with $c_j=1$, $L_j=R(s_j)^{\dagger}Z_{T_0\cap C_j}$, and $\ket{f_j}$ the signed sum of the local
environment tensors over the local fibre. The signs distribute because the eigenvalue of a stabilizer on
a fixed syndrome sector is a character of the stabilizer group, so
$\epsilon_{\prod_jg_j}(s_{\mathrm{in}})=\prod_j\epsilon_{g_j}(s_{\mathrm{in}})$ and each component's sign
belongs to its own $\ket{f_j}$, with no residual global factor. Terminality of the deviation environments
is a hypothesis of the setting and is not derived here.

\emph{The accepted fibre.} On the fibre $s=s_{\mathrm{in}}$ the surviving $T$ have
$\operatorname{syn}(Z_T)=0$, so by the same check-separation argument each $Z_{T\cap C_j}$ has trivial
syndrome and weight below $d_Z$, hence is a stabilizer, hence $Z_T$ is one. Here $R_{h,s}=R(0)=I$ and
every $L_E$ is trivial, so the ambiguity group on this fibre is $\{I\}$ and \cref{eq:absorb} holds for any
following free instrument. This is the fibre \cref{def:splitreadout} ends with a native $\Xbar$
measurement, and $M^X_x=\tfrac12(I+x\Xbar)$ commutes with every stabilizer, so it carries each syndrome
sector into itself and the sector-preservation clause of \cref{def:absorption} holds on this fibre.

\emph{The rejected fibres.} On a fibre with $s\neq s_{\mathrm{in}}$ the ambiguity group is $\{I,\Zbar\}$ as
above, and \cref{def:splitreadout} applies the Pauli correction $C=R(s-s_{\mathrm{in}})$ and then measures
$\Zbar$ natively. For $L\in\{I,\Zbar\}$ we have $CR_{h,s}L=\pm LCR_{h,s}$ and
$M^Z_zL=z^{[L=\Zbar]}M^Z_z$, so $M^Z_zC\,R_{h,s}L\,J_hP=\pm z^{[L=\Zbar]}\,M^Z_zC\,R_{h,s}J_hP$ for any
incoming branch, which is \cref{eq:absorb} as it is stated. The Paulis $R_{h,s}$ and $C$ carry each
syndrome sector onto another syndrome sector, which is what a correction does, and
$M^Z_z=\tfrac12(I+z\Zbar)$ commutes with every stabilizer and so carries each sector into itself.

Both fibres therefore discharge \cref{def:absorption} at $D=d_Z$. With \cref{cond:components} at $d_Z$ in hand and
\cref{cond:skeleton} holding for this skeleton, \cref{prop:absorption} gives \cref{eq:h4} at the threshold
$D=d_Z$, so $d_Z$ is a certified threshold and $\Dins\ge d_Z$.
\end{proof}

\begin{proof}[Proof of \cref{prop:family-magic}]
On the first slab $J_h=I$ and $s_{\mathrm{in}}=0$. With unitary marked cells the layer is
$U=\prod_x(\cos\theta_xI+i\sin\theta_xZ_{q(x)})=\sum_{T}c_TZ_T$ with
$c_T=\prod_{x\in T}(i\sin\theta_x)\prod_{x\notin T}\cos\theta_x$, and the environment is one-dimensional so
no trace remains to be taken. Projecting onto the zero syndrome keeps the $T$ with
$\operatorname{syn}(Z_T)=0$, that is those for which $Z_T$ lies in the $Z$-type normaliser, and each such
$Z_T$ acts on the code space as $I$ or as $\Zbar$ according to its class. Writing $\alpha$ for the sum of
the $c_T$ of the first kind and $\beta$ for the sum of the second,
$\Pi_0UP=(\alpha I+\beta\Zbar)P$. Since $\Xbar$ commutes with every stabilizer, $M^X_+$ commutes with
$\Pi_0$, and the accepted effect is
$F_A=PU^{\dagger}\Pi_0M^X_+\Pi_0UP=(\alpha I+\beta\Zbar)^{\dagger}M^X_+(\alpha I+\beta\Zbar)$, which is
\cref{eq:family-effect}.

Expanding with $\Zbar\Xbar\Zbar=-\Xbar$ and $\Xbar\Zbar=-i\Ybar$,
\begin{equation}\label{eq:family-bloch}
  F_A=\tfrac12\Bigl[(a+b)I+(a-b)\Xbar+2\operatorname{Im}(\bar\alpha\beta)\,\Ybar
  +2\operatorname{Re}(\bar\alpha\beta)\,\Zbar\Bigr],\qquad a=|\alpha|^2,\ b=|\beta|^2 .
\end{equation}
The normalised Bloch vector has
$\|\bm v\|_1=\bigl[\,|a-b|+2|\operatorname{Re}(\bar\alpha\beta)|+2|\operatorname{Im}(\bar\alpha\beta)|\,\bigr]/(a+b)$,
and $|\operatorname{Re}w|+|\operatorname{Im}w|\ge|w|$ gives
$\|\bm v\|_1\ge(|a-b|+2\sqrt{ab})/(a+b)$. For $a>b>0$ the numerator exceeds the denominator by
$2\sqrt b(\sqrt a-\sqrt b)>0$, and for $b>a>0$ by $2\sqrt a(\sqrt b-\sqrt a)>0$, so $\Dstab(\Ehat)>0$
whenever $\alpha\beta\neq0$ and $|\alpha|\neq|\beta|$.

Finally $\beta$ is a sum over $T\subseteq\mathcal R_d$ with $Z_T$ logical. If no subset of the marked set
carries a $Z$-type logical operator the sum is empty and $\beta=0$, so $\beta\neq0$ requires the marked
set to contain the support of one, which has weight at least $d_Z$.

For the converse, let the marked set contain such a support and let the angles run over
$(-\pi/2,\pi/2)$, so that every $\cos\theta_x\neq0$. Put $u_x=\tan\theta_x$ and factor
$\prod_x\cos\theta_x$ out of both sums, which leaves
\begin{equation}\label{eq:generic-poly}
  \alpha=\Bigl(\prod_x\cos\theta_x\Bigr)\,\tilde\alpha(u),\qquad
  \beta=\Bigl(\prod_x\cos\theta_x\Bigr)\,\tilde\beta(u),\qquad
  \tilde\alpha(u)=\!\!\sum_{T:\,Z_T\in\mathcal G}\!\! i^{|T|}u^T,\quad
  \tilde\beta(u)=\!\!\sum_{T:\,Z_T\ \mathrm{logical}}\!\! i^{|T|}u^T,
\end{equation}
with $u^T=\prod_{x\in T}u_x$. Distinct $T$ give distinct monomials, so $\tilde\alpha$ and $\tilde\beta$
are polynomials in $u$ whose coefficients are the $i^{|T|}$ and are never zero. Hence $\tilde\alpha$ is not
the zero polynomial, since $T=\varnothing$ contributes the constant $1$, and $\tilde\beta$ is not the zero
polynomial exactly when some logical $T$ is available, which is the hypothesis. The real polynomial
$|\tilde\alpha|^2-|\tilde\beta|^2$ in the real variables $u$ is not identically zero either, since it
takes the value $1$ at $u=0$. Each of $\alpha=0$, $\beta=0$ and $|\alpha|=|\beta|$ therefore confines $u$
to the zero set of a polynomial that does not vanish identically, which is closed and of measure zero, and
$\theta\mapsto u$ is a diffeomorphism of $(-\pi/2,\pi/2)^{|\mathcal R_d|}$ onto $\mathbb R^{|\mathcal
R_d|}$. The angles at which the accepted effect fails to carry magic are therefore contained in a closed
set of measure zero, closed relative to that cube.
\end{proof}

\Cref{eq:family-effect,eq:family-bloch} were checked against a direct state-vector simulation of the whole
protocol on the distance-3 rotated surface code, at four marked sets and random angles, agreeing to
$10^{-15}$. The marked set consisting of two qubits, which contains no logical support, gave
$\beta=0$ and $\Dstab=0$; the minimum-weight logical column gave $\Dstab=3.3\times10^{-2}$; and marking
all nine qubits gave $\Dstab=0.38$, so the family is not one whose accepted effect is free. The
logical column at $\theta_x=\pi/4$, where $|\alpha|=|\beta|$, gave $\Dstab=0$ exactly, which is the
edge the positivity criterion excludes. Over $3000$ random angle vectors on that same marked set, drawn
uniformly from $(-\pi/2,\pi/2)$, none gave zero magic, the smallest value seen being $1.0\times10^{-5}$,
which is what the measure-zero claim predicts.

The two size-dependent steps were checked directly on the rotated surface code at $d=3$ and $d=5$, which
the accompanying scripts first validate to be $[[d^2,1,d]]$ codes with $d_X=d_Z=d$ and commuting checks.
Within one configuration, on a random sample of $300$ separated configurations at each distance, every
syndrome fibre with more than one member had all its members equal modulo $Z$-stabilizers, over $438$ and
$11\,082$ such fibres. The tightness of \cref{rem:standard-tight} was confirmed by exhibiting a
weight-$d$ column as a $Z$-type logical. For the ambiguity group, containment in $\{I,\Zbar\}$ is a rank
statement and not a sampling question, since $\dim\ker H_X-\operatorname{rank}H_Z=1$ leaves exactly two
classes, and the scripts print those ranks. What the sample shows is that the group is not trivial: over
$21$ and $120$ sampled configurations, syndrome fibres containing terms from two different configurations
supplied $323$ and $19\,518$ pairs differing by a stabilizer and $34$ and $422$ pairs differing by
$\Zbar$.

\subsection{The weak-string protocol}\label{app:localclass-construction}

This subsection proves \cref{thm:construction}. Throughout, $d$ is odd, $\gamma=\{q_1,\dots,q_d\}$ is a
bare non-self-intersecting minimum-weight logical $\Zbar$ string, $\theta_d=\kappa/d$ with
$0<\kappa<1$, and $c=\cos\theta_d$, $s=\sin\theta_d$. The protocol is \cref{def:weakstring}. The free
skeleton is the identity at every marked cell, so
\begin{equation}\label{eq:weak-delta}
  \Delta_j\;=\;U_j-I\;=\;(c-1)I+is\,Z_j ,
  \qquad \eta_j=\|U_j-I\|=2\sin(\theta_d/2)=\frac{\kappa}{d}+O(d^{-3}),
\end{equation}
and $\varepsilon_j\le\|\mathcal U_{\theta_d}-\id\|_\diamond=2\sin\theta_d=\Theta(1/d)$. Both cells are
unitary with one-dimensional environments, so there is no dilation gauge freedom and
\cref{rem:amplitude} applies with $\eta_j=\Theta(\varepsilon_j)$.

\paragraph{The kernel property.} If a nonempty $R\subsetneq\gamma$ had trivial syndrome then $Z_R$ would
lie in the normaliser. If $Z_R$ were a nontrivial logical operator its weight would be below $d$, and if
$Z_R$ were a stabilizer then $Z_\gamma Z_R$ would represent the same logical operator on
$\gamma\setminus R$, again of weight below $d$. Both contradict the code distance, so
\begin{equation}\label{eq:kernel}
  \operatorname{syn}(Z_T)=\operatorname{syn}(Z_{T'})\ \text{ for }T,T'\subseteq\gamma
  \quad\Longleftrightarrow\quad T\bigtriangleup T'\in\{\varnothing,\gamma\} .
\end{equation}
A syndrome fibre inside $\gamma$ is therefore a single support or a complementary pair.

\paragraph{The setting and the class hypotheses.} \Cref{cond:one-block}, which fixes the setting rather
than entering any bound, holds because there is one block and the check
ancillas are fresh local stabilizer ancillas. \Cref{cond:skeleton} holds because the marked layer has
depth one and no marked cell propagates beyond its own data qubit before the syndrome round, whose
outcomes are stored in orthogonal registers at the raw-record cut. For \cref{cond:degree}, take
$\Gfrak_d$ to carry the path adjacency $q_j\sim q_{j+1}$ together with the bounded-range edges from
overlapping check gadgets. Each vertex acquires $O(1)$ edges, so the degree is bounded, while
$\Gfrak_d[\gamma]$ is connected, which is what makes $\gamma$ a single component of size $d$. Decoding,
the test of $s=0$ and the terminal measurement all happen after the cut and contribute no edges, which
is the point of stopping the light cones there.

\paragraph{Branch factorisation and the threshold.} We verify \cref{def:absorption} and
\cref{cond:components}, and appeal to \cref{prop:absorption}. \Cref{cond:components} is immediate for this
protocol, because the marked cells are single-qubit rotations on distinct data qubits in one layer, so the dressed support of
a component is the set of check gadgets touching its qubits, and components that are non-adjacent in $\Gfrak_d$ touch
disjoint gadget sets. The syndrome of $Z_T$ is the pair of stabilizers at the two ends of each run of $T$, so it
splits over components, the insertion monomial $\prod_{j\in S}\Delta_j$ is a product over components by construction,
and the classical processing past the cut is the test of $s=0$ together with the syndrome-dependent
correction $R(s)$ on the rejected branch. The test reads one bit of the whole syndrome, and $R$ is the
fixed linear section, so neither builds a coherent function of two components. For an insertion set $S\subseteq\gamma$,
$\prod_{j\in S}\Delta_j=\sum_{T\subseteq S}(c-1)^{|S|-|T|}(is)^{|T|}Z_T$, and projecting onto syndrome
$s_0$ retains only the $T$ with $\operatorname{syn}(Z_T)=s_0$. By \cref{eq:kernel} the ambiguity group is
contained in $\{I,\Zbar\}$. For $s_0\ne0$ the two representatives differ by $\Zbar$, and the correction $R(s_0)$ followed by the
terminal $\Zbar$ measurement absorbs it through $M_z^Z\Zbar=zM_z^Z$, which is \cref{eq:absorb} with
$\chi(\Zbar)=\pm z$. For $s_0=0$ and $S\subsetneq\gamma$ the only syndrome-free $T\subseteq S$ is
$T=\varnothing$, so the branch is $(c-1)^{|S|}M_x^XP$ and the ambiguity group is trivial. Hence every
insertion set whose components are smaller than $d$ lies on one free ray, so $d$ is a certified
threshold. At $S=\gamma$ the factorisation fails, since
\begin{equation}\label{eq:full-gamma}
  M_+^X\Bigl[(c-1)^dI+(is)^d\Zbar\Bigr]P
\end{equation}
has two summands that are not proportional, $M_+^X$ and $M_+^X\Zbar$ being independent. Note that
\cref{eq:full-gamma} is built from the insertion monomial $\prod_j\Delta_j$ and not from the full layer
$U_\gamma$, whose coefficients would be $c^d$ and $(is)^d$. Both have non-proportional parts, but the
threshold is a statement about insertion sets, so it is the former that is relevant. Since $\gamma$ is a
single component of size $d$ this is the first failure, giving $\Dins=d$ and \cref{cond:dins}.

\paragraph{The accepted effect.} By \cref{eq:kernel} the zero-syndrome branch of the full layer is
$PU_\gamma P=c^dI+(is)^d\Zbar$. Writing $\alpha=c^d$ and $\beta=\pm s^d$ for odd $d$, the zero-syndrome
logical operation is $A=\alpha I+i\beta\Zbar$, and accepting the $+1$ outcome of $\Xbar$ gives
$F_A=A^\dagger\tfrac12(I+\Xbar)A$. Using $\Zbar\Xbar\Zbar=-\Xbar$ and $\Xbar\Zbar=-i\Ybar$,
\begin{equation*}
  A^\dagger\Xbar A=(\alpha^2-\beta^2)\Xbar+2\alpha\beta\,\Ybar,
  \qquad A^\dagger A=(\alpha^2+\beta^2)I ,
\end{equation*}
which is \cref{eq:constr-effect}. Then $\pacc=\tfrac12\Tr F_A=(\alpha^2+\beta^2)/2$, and with
$t=|\beta|/\alpha=\tan^d\theta_d$ the normalised Bloch vector is
$\bigl((1-t^2)/(1+t^2),\ \pm2t/(1+t^2),\ 0\bigr)$, whose $\ell_1$ norm exceeds one by
$2t(1-t)/(1+t^2)$ for $0<t<1$. Multiplying by $\pacc$ and simplifying gives
$\magicc(F_A)=\alpha^2t(1-t)=c^ds^d-s^{2d}$, which is \cref{eq:constr-magic}.

\paragraph{Asymptotics.} With $\theta_d=\kappa/d$ we have $t=\tan^d(\kappa/d)=(\kappa/d)^d[1+o(1)]$ and
$c^{2d}=1-O(1/d)$, while $s^{2d}$ is superexponentially small, so $\pacc=\tfrac12-O(1/d)$ and
$\magicc(F_A)=(\kappa/d)^d[1+o(1)]$, which is \cref{eq:constr-asymp}. On the other side,
\cref{eq:lambda-count} with $m=d$, $\Dins=d$ and $\eta_{\max}=\kappa/d+O(d^{-3})$ gives
$\Lambda\le d(C_\zeta\kappa/d)^d(1+o(1))$, so \cref{thm:suppression} yields
$-\log\magicc(F_A)\ge d\log d-O(d)$. The two agree in the coefficient of $d\log d$, so the bound is
attained at leading order on this family. The agreement is in the amplitude form of the theorem. Passing
through the generic $\eta=O(\sqrt\varepsilon)$ estimate would give the weaker exponent
$\tfrac12d\log d$, and the family does not show that weaker exponent to be attained.